\documentclass[11pt]{article} 
\usepackage{fullpage}
\usepackage{graphicx}
\usepackage{latexsym}

\usepackage{multirow}
\usepackage{mdframed}
\usepackage{tikz}
\usepackage{color} 
\usepackage[linesnumbered,ruled]{algorithm2e} 

\usepackage[toc,page]{appendix}

\usepackage{amsmath, amsthm, amssymb}           
\usepackage[english]{babel}                                  
\usepackage[T1]{fontenc}                                      
\usepackage[utf8]{inputenc}
\usepackage[round]{natbib}
\usepackage{array}           
\usepackage{geometry}
\usepackage{subcaption}
\usepackage{bbm}
\usepackage{hyperref}
\usepackage{dsfont}
\usepackage{mathtools} 
\usepackage{comment}
\usepackage{stackrel}
\usepackage{pdfpages}
\usepackage{extarrows} 
\usepackage{listings} 
\usepackage{ulem} 
\usepackage{hyphenat} 
\usepackage[toc,page]{appendix}
\usepackage{footmisc}
\usepackage{authblk}
\usepackage{thmtools} 
\newcolumntype{P}[1]{>{\centering\arraybackslash}p{#1}}
\newcolumntype{M}[1]{>{\centering\arraybackslash}m{#1}}
\newcolumntype{L}[1]{>{\raggedright\arraybackslash}m{#1}}
\newcolumntype{C}[1]{>{\centering\arraybackslash}m{#1}}
\newcolumntype{R}[1]{>{\raggedleft\arraybackslash}m{#1}}

\definecolor{upmaroon}{rgb}{0.48, 0.07, 0.07}
\definecolor{royalazure}{rgb}{0.0, 0.22, 0.66}
\definecolor{pakistangreen}{rgb}{0.0, 0.4, 0.0}

\theoremstyle{plain}                                
\newtheorem{theorem}{Theorem}[section]
\newtheorem{defn}[theorem]{Definition}

\newtheorem{exmpl}[theorem]{Example}   

\newtheorem{lemma}[theorem]{Lemma}
\newtheorem{assumptions}[subsection]{Assumption}

\usepackage{enumitem}
\setlist[itemize]{noitemsep, topsep=0pt}

\def\iid{independent and identically distributed}

\newcommand{\eps}{\varepsilon}

\newcommand{\R}{\mathbb{R}} 

\newcommand{\one}{\mathds{1}}

\newcommand\independent{\protect\mathpalette{\protect\independenT}{\perp}}
\def\independenT#1#2{\mathrel{\rlap{$#1#2$}\mkern2mu{#1#2}}}

\newcommand{\norm}[1]{\lVert #1 \rVert}
\newcommand{\normbig}[1]{\big\lVert #1 \big\rVert}
\newcommand{\normBig}[1]{\Big\lVert #1 \Big\rVert}

\newcommand{\normP}[2]{\norm{#1}_{\PP, #2}}
\newcommand{\normPbig}[2]{\normbig{#1}_{\PP, #2}}
\newcommand{\normPBig}[2]{\normBig{#1}_{\PP, #2}}
\newcommand{\normone}[1]{\lvert #1\rvert}
\newcommand{\normonebig}[1]{\big\lvert #1\big\rvert}
\newcommand{\normoneBig}[1]{\Big\lvert #1\Big\rvert}
\newcommand{\normonebigg}[1]{\bigg\lvert #1\bigg\rvert}

\DeclareMathOperator{\E}{\mathbb{E}}

\DeclareMathOperator{\Cov}{Cov}
\DeclareMathOperator{\Prob}{\mathbb{P}}
\DeclareMathOperator{\Var}{Var}

\DeclareMathOperator{\median}{median}
\DeclareMathOperator{\CI}{CI}
\newcommand{\Ber}{\mathrm{Bernoulli}}

\newcommand{\NN}{N}
\newcommand{\nn}{n}
\newcommand{\KK}{K}
\newcommand{\kk}{k}

\newcommand{\TauN}{\mathcal{T}} 

\newcommand{\PP}{P}
\newcommand{\PcalN}{\mathcal{P}_{\NN}}

\newcommand{\deltaN}{\delta_{\NN}}
\newcommand{\DeltaN}{\Delta_{\NN}}

\newcommand{\indset}[1]{[#1]}

\newcommand{\salg}{b}

\newcommand{\Salg}{B}

\newcommand{\W}{\boldsymbol{W}}
\newcommand{\Yi}{Y_i}

\renewcommand{\Xi}{X_i}

\newcommand{\Xj}{X_j}
\newcommand{\Zi}{Z_i}

\newcommand{\Zj}{Z_j}
\newcommand{\Wi}{W_i}
\newcommand{\Wj}{W_j}
\newcommand{\Wminusi}{W_{-i}}

\newcommand{\Wm}{W_m}

\newcommand{\Si}{S_i}
\newcommand{\Sj}{S_j}
\newcommand{\Sm}{S_m}
\newcommand{\Sr}{S_r}

\renewcommand{\S}{S}
\newcommand{\Ci}{C_i}
\newcommand{\Cminusi}{C_{-i}}
\newcommand{\Cj}{C_j}
\newcommand{\Cm}{C_m}

\newcommand{\CXZi}{\Ci, \Xi, \Zi}
\newcommand{\CXi}{\Ci, \Xi}
\newcommand{\CXione}{\Ci, \Xi^{\pi}}
\newcommand{\CXizero}{\Ci, \Xi^{1-\pi}}
\newcommand{\CZi}{\Ci, \Zi}
\newcommand{\CZj}{\Cj, \Zj}
\newcommand{\Di}{D_i}
\newcommand{\Ui}{U_i}
\newcommand{\Xipi}{\Xi^{\pi}}
\newcommand{\Ximinuspi}{\Xi^{1-\pi}}
\newcommand{\Zipi}{\Zi^{\pi}}
\newcommand{\Ziminuspi}{\Zi^{1-\pi}}

\newcommand{\thetazeroN}{\theta_{\NN}^0}
\newcommand{\thetazeroi}{\theta_i^0}
\newcommand{\thetazeroj}{\theta_j^0}
\newcommand{\thetazerod}{\theta_{d}^0}
\newcommand{\thetazerodegr}[1]{\theta_{\degr{#1}}^0}
\newcommand{\htheta}{\hat\theta}
\newcommand{\hthetadegr}[1]{\hat\theta_{\degr{#1}}}
\newcommand{\hthetad}{\hat\theta_{d}}

\newcommand{\hthetas}{\htheta_{\salg}}

\newcommand{\thetazeroNGlobal}{\xi_{\NN}^0}
\newcommand{\thetazeroiGlobal}{\xi_i^0}

\newcommand{\hthetaGlobal}{\hat\xi}

\newcommand{\hthetadGlobal}{\hat\xi_{d}}

\newcommand{\SIk}{\mathcal{S}_{\Ik}}
\newcommand{\SIkc}{\mathcal{S}_{\Ikc}}
\newcommand{\SIkcdegr}[1]{\mathcal{S}_{I_{k(#1)}^c}}
\newcommand{\SIkdegr}[1]{\mathcal{S}_{I_{k(#1)}}}
\newcommand{\Ikdegr}[1]{I_{k(#1)}}

\newcommand{\gzerozero}{g_0^0}
\newcommand{\gonezero}{g_1^0}
\newcommand{\hzero}{h^0}

\newcommand{\gzero}{g_0}
\newcommand{\gone}{g_1}
\newcommand{\goneIkc}{\hat g_1^{\Ikc}}
\newcommand{\gzeroIkc}{\hat g_0^{\Ikc}}
\newcommand{\h}{h}
\newcommand{\hIkc}{\hat h^{\Ikc}}
\newcommand{\eIkc}{\hat e^{\Ikc}}
\newcommand{\fk}{f_k}

\newcommand{\etazero}{\eta^0}
\newcommand{\hetaIkc}{\hat\eta^{\Ikc}}
\newcommand{\hetaIkcdegr}[1]{\hat\eta^{I_{\kk(#1)}^c}}

\newcommand{\Ik}{I_{\kk}}
\newcommand{\Ikc}{I_{\kk}^c}

\newcommand{\epsYi}{\eps_{\Yi}}
\newcommand{\epsY}{\eps_{Y}}
\newcommand{\epsWi}{\eps_{\Wi}}
\newcommand{\epsCi}{\eps_{\Ci}}

\renewcommand{\phi}{\varphi}
\newcommand{\dmax}{d_{\mathrm{max}}}
\newcommand{\tildedmax}{ d_{\mathrm{max}}'}
\newcommand{\degr}[1]{d(#1)}

\newcommand{\Acald}{\mathcal{A}_{d}}

\newcommand{\tildeG}{G_D}

\newcommand{\tildeE}{E_D}

\newcommand{\Cnormp}{C_1}
\newcommand{\Cnorminfty}{C_2}
\newcommand{\CnormEta}{C_5}
\newcommand{\CproductAndDegre}{C_6}
\newcommand{\Chzero}{C_3}
\newcommand{\Ctheta}{C_4}
\newcommand{\CLpfournorm}{C_7}

\newcommand{\sigmainfty}{\sigma_{\infty}}

\newcommand{\hvars}{\hat\sigma_{\infty, \salg}^2}
\newcommand{\hsigmas}{\hat\sigma_{\infty, \salg}}
\newcommand{\sigmaN}{\sigma_{\NN}}
\newcommand{\ps}{p_{\salg}}
\newcommand{\paggr}{p_{\mathrm{aggr}}}
\newcommand{\paggrzero}{\paggr^0}
\newcommand{\paggrTheta}{p_{\mathrm{aggr}}^{\theta}}

\newcommand{\Peps}{P_{\varepsilon}}
\newcommand{\PC}{P_{C}}
\newcommand{\T}{\sigmainfty^2(\Peps, \PC, \etazero; G)}
\newcommand{\hPheps}{\hat{P}_{\hat\varepsilon}}
\newcommand{\hPC}{\hat{P}_{C}}
\newcommand{\heta}{\hat\eta}
\newcommand{\Tboot}{\sigmainfty^2(\hPheps, \hPC, \heta; G)}
\newcommand{\hthetazeroi}{\hat{\theta}^0_i}
\newcommand{\hepsYi}{\hat{\varepsilon}_{Y_i}}
\newcommand{\hepsYj}{\hat{\varepsilon}_{Y_j}}
\newcommand{\gzeroIkci}[1]{\hat{g}_0^{I^c_{k(#1)}}}
\newcommand{\goneIkci}[1]{\hat{g}_1^{I^c_{k(#1)}}}
\newcommand{\hIkci}[1]{\hat{h}^{I^c_{k(#1)}}}

\DeclareMathOperator{\pa}{pa}

\begin{document}

\title{Treatment Effect Estimation with Observational Network Data using Machine Learning}  
  
\author[1]{Corinne Emmenegger} 
\author[1]{Meta-Lina Spohn}
\author[2]{Timon Elmer}
\author[1]{Peter B\"uhlmann}
\affil[1]{Seminar for Statistics, ETH Z\"urich}
\affil[2]{Department of Humanities, Social and Political Sciences, ETH Z\"urich}

\date{}

\setcounter{Maxaffil}{0}
\renewcommand\Affilfont{\itshape\small}
\maketitle

\begin{abstract}
\noindent
Causal inference methods for treatment effect estimation usually assume independent units. However, this assumption is often questionable because units may interact, resulting in spillover effects between them.
We develop augmented inverse probability weighting (AIPW) for estimation and inference of the expected average treatment effect (EATE) with observational data from a single (social) network with spillover effects. In contrast to overall effects such as the global average treatment effect (GATE), the EATE measures, in expectation and on average over all units, how the outcome of a unit is causally affected by its own treatment, marginalizing over the spillover effects from other units. We develop cross-fitting theory with plugin machine learning to obtain a semiparametric treatment effect estimator that converges at the parametric rate and asymptotically follows a Gaussian distribution. The asymptotics are developed using
the dependency graph rather than the network graph, which makes explicit that we allow for spillover effects beyond immediate neighbors in the network. We apply our AIPW method to the Swiss StudentLife Study data to investigate the effect of hours spent studying on exam performance accounting for the students' social network.
\end{abstract}


\section{Introduction}\label{sec:intro}

Classical causal inference approaches for treatment effect estimation with observational data usually assume independent units.
This assumption is part of the common stable unit treatment value assumption (SUTVA)~\citep{rubin1980}. However, independence is often violated  
in practice due to interactions among units that lead to so-called spillover effects.
For example, the vaccination against an infectious disease (treatment) of a person (unit) may not only influence this person's health status (outcome), but may also protect the health status of other people the person is interacting with~\citep{Perez-Heydrich2014, Saevje2021}. 
In the presence of spillover effects, standard algorithms fail to separate correlation from causation, and spurious associations due to network dependence contribute to the replication crisis~\citep{Lee-Ogburn2021} and may yield biased causal effect estimators and invalid inference~\citep{Sobel2006, Perez-Heydrich2014, sofrygin2017, eckles2021, Lee-Ogburn2021, ogburn2014b}.
New approaches are required to guarantee valid causal inference from observational data with spillover effects.

We consider the following types of spillover effects: i) causal effects of other units' treatments on a given unit's outcome, referred to as interference in the literature~\citep{Sobel2006, Hudgens-Halloran2008}, and ii) causal effects of other units' covariates on a given unit's treatment or outcome\footnote{Another notion of spillover effects is frequently used in the social sciences; please see Section~\ref{sect:app_socsci} in the appendix for a discussion.}. The spillover effects a unit receives are governed by proximity of this unit to other units in a known undirected network $G$.
The edges of this network represent some kind of interaction or relationship of the respective units such as friendship, geographical closeness, or shared department in a company.

In this paper, the causal effect of interest and target of inference is the expected average treatment effect (EATE)~\citep{Saevje2021} in an observational setting. The EATE measures, in expectation and on average over all units, how the outcome of a unit is causally affected by its own treatment in the presence of spillover effects from other units. The EATE is the statistical parameter when the question is how, on average for all units, the outcome of a specific unit is influenced when only its own treatment is altered. In the infectious disease example, the EATE measures the average expected difference in health status of an individual assigned to the vaccination versus not, 
marginalizing over unit-specific covariates and spillover effects of other people. This corresponds to the medical effect of the vaccine in a person's body. 
This interpretation highlights that the EATE is not an estimand for policy evaluation, where, for example, one is interested in capturing 
the effect of jointly vaccinating a sample of the population.

We now formalize the EATE following~\citep{laan2017}. 
For each unit $i=1,2,\ldots,\NN$, let $\Wi\in\{0, 1\}$ be the dichotomous treatment, $\Yi$ be the response, and $\Ci$ be the covariates of unit $i$. The $\NN$ units are connected in a fixed undirected network $G$ in which they may exhibit spillover effects of the two above mentioned types $i)$ and $ii)$ from their immediate neighbors and/or units further away. Let $P_L^N$ be the observational distribution of $O = (W_i, C_i, Y_i)_{i = 1, \ldots, N}$, where $L$ is the distribution of $W = (W_1, \ldots, W_N) $ given $C = (C_1, \ldots, C_N)$. Let $P_{\tilde{L}}^N$ be the distribution of $\widetilde{O} = (\widetilde{W_i}, C_i, \widetilde{Y_i})_{i = 1, \ldots, N}$, where the conditional distribution $L$ of $W$ given $C$ has been replaced by the user defined distribution $\tilde{L}$. This distribution $\tilde{L}$ describes the intervention on the treatment vector $W$ that the researcher is interested in. We can then define the EATE as  
\begin{align*}
    \theta_N^0 = \theta_N^0(1) - \theta_N^0(0), 
 \end{align*} where 
 \begin{align*}
    \theta_N^0(w):= \frac{1}{N}\sum_{i =1}^N \E_{P_{\tilde{L}_i(w)}^N}\left[Y_i^{\text{do}(W = \tilde{L}_i(w))}\right],   
 \end{align*} where we use the do-notation of~\citet{pearl1995} and
\begin{align*}
    \tilde{L}_i(w) = (W_1, \ldots, W_{i-1}, w, W_{i+1}, \ldots, W_N), \quad \text{for} \quad w \in \{0, 1\},
\end{align*} represents the intervention on the unit-specific treatment $W_i$ (setting it to constant $w$), but the distribution of treatments $W_j$ for the other $N-1$ units $j$ in the network are left unchanged. In particular, the intervention $\tilde{L}_i(w) $ is independent of $C$. Thus, $ \theta_N^0(1) $ evaluates a collection of unit-specific distributions, $(\tilde{L}_1(1), \ldots, \tilde{L}_N(1))$, which cannot be rewritten as a single intervention $\tilde{L}$ on the whole treatment vector $W$. By denoting the EATE by $ \theta_N^0$, it remains implicit that it is defined conditional on a specific network $G$, while it is explicit that it is a function of the given sample size $N$. Consequently, the EATE's true value can vary depending on the sample size and the network structure.

To simplify notation, we rewrite the EATE by 
\begin{displaymath}
	\thetazeroN = \frac{1}{\NN}\sum_{i=1}^{\NN} \thetazeroi, 
\end{displaymath}
where $$\thetazeroi = \E_{W_{-i}, C_{-i}, C_i}\Big[ \E\left[\Yi^{do(\Wi = 1)} - \Yi^{do(\Wi = 0)}\mid W_{-i}, C_{-i}, C_i\right]\Big], $$ and $W_{-i} = (W_1, \ldots, W_{i-1}, W_{i+1}, \ldots, W_N)$ and $C_{-i} = (C_1, \ldots, C_{i-1}, C_{i+1}, \ldots, C_N)$. Thus, the EATE equals the average of the unit-specific treatment effects $\thetazeroi$, that is, the expected difference in outcomes $\Yi$ if the treatment was assigned to unit $i$ versus if it was retained from unit $i$. The unit-specific treatment effects may not be the same for all units because the outcomes may have different distributions conditional on $W_{-i}$ and $C_{-i}$ across units due to the spillover effects. In the setting without spillover effects, the distribution of $Y_i$ does not depend on $W_{-i}$ and $C_{-i}$, for each $i = 1, \ldots, N$, and thus the EATE coincides with the average treatment effect (ATE) if spillover effects are absent~\citep{Neyman1923, Rubin1974}.

We impose the following key assumption (that is standard in this literature~\citep{laan2017,Laan2014, sofrygin2017}): 
the spillover effects can be summarized by lower dimensional features. That is, we will use domain knowledge-informed features that are arbitrary functions of the network $G$ and the treatment and covariate vectors of all units~\citep{manski1993, chin2019}.
The features are assumed to capture all pathways through which spillover effects take place. 
For example, \citet{Cai2015} and~\citet{Leung2020} model the purchase of a weather insurance (outcome) of farmers in rural China as a function of attending a training session (treatment) and the proportion of friends who attend the session (feature on direct neighbors in the network).

In the following, we will assume a structural equation model (SEM) to impose our assumptions on the data generating mechanism of the joint distribution of $(W_i, C_i, Y_i)_{i =1\ldots, N}$. The outcome and propensity score model of the SEM may be highly complex and nonsmooth and include interactions and high-dimensional variables. 
We then follow an augmented inverse probability weighting (AIPW)~\citep{Robins1995} approach to estimate the EATE $\thetazeroN$ in the context of this model. 
We estimate the outcome and propensity score models with arbitrary machine learning algorithms and plug them into our AIPW estimand identifying $\thetazeroN$. 
These machine learning estimators may be highly complex and  
suffer from regularization bias and slow converge rates. However, the use of sample splitting with cross-fitting~\citep{Chernozhukov2018} allows us to address these issues. 
Limiting the growth of dependencies between units, our estimator of the EATE is consistent, converges at the $\sqrt{N}$-rate, and asymptotically follows a Gaussian distribution. This allows us to construct confidence intervals and p-values.

\subsection{Our Contribution and Comparison to Literature}\label{sect:contribution}

Our work is most related to the 
literature on semiparametric treatment effect estimation and inference with observational data from a single network. 
\citet{autoG2021} develop a network version of the g-formula~\citep{Robins1986} and perform outcome regression, assuming that the data can be represented as a chain graph, which is a graphical model that is generally incompatible with our SEM approach~\citep{lauritzen2002}.
An SEM approach is also used by \citet{Laan2014}, \citet{laan2017} and \citet{sofrygin2017}.
These works consider a similar model as we do and propose semiparametric treatment effect estimation by targeted maximum likelihood (TMLE)~\citep{vanderLaanRubin2006, LaanBook2011, LaanBook2018}. 
\citet{Laan2014} and~\citet{sofrygin2017} 
primarily consider global effects that  compare two hypothetical interventions on the whole treatment vector. An example of such an effect is 
the global average treatment effect (GATE), which contrasts the interventions of treating all units of the population versus treating no unit of the population. In contrast, we consider the EATE that is the average effect of assigning the treatment to one unit versus not and integrate out the treatment selections from the other units. Causal effects like the EATE summarizing the effect of $\NN$ unit-specific interventions generally cannot be described using a single intervention on the whole treatment vector, as done for global effects. The behavior of estimators for the EATE under the wrong i.i.d. assumption is studied by~\citep{Saevje2021}.  \citet{laan2017} mention a possible extension to estimate the EATE with TMLE, but all their results are for global effects such as the GATE. 
Their theory assumes some kind of a bounded entropy integral, 
which is difficult to verify for machine learning methods.

Our contribution includes the following. 
First, 
we present a semiparametric, machine learning-based approach to estimate the EATE with observational data from a single network.
Our approach enables performing inference, including confidence intervals and p-values. 
Particularly, we do not require multiple disjoint networks. 
We develop a cross-fitting algorithm under interference and reason in terms of the dependency graph to explicitly allow for different interactions, also specifically ones that are beyond immediate neighbors in the network.
Second, the limiting asymptotic Gaussian distribution and optimal $\sqrt{\NN}$-convergence rate of the EATE estimator are achieved even if the number of ties of a unit may diverge asymptotically. 
To reach this optimal convergence rate to estimate global effects, \citet{sofrygin2017} need to uniformly bound the neighborhood size of a unit.
Third, our algorithm based on sample splitting is easy to understand and implement, and the user may choose any machine learning algorithm.
Fourth, we analyze the Swiss StudentLife Study data~\citep{Stadtfeld2019, Voeroes2021} and estimate the effect of study time on the grade point average of freshmen students after their first-year examinations at one of the world's leading universities.

\textit{Outline of the Paper.}
Section~\ref{sect:estimator} presents the model assumptions, characterizes the treatment effect of interest, outlines the procedures for the point estimation of the EATE and estimation of its variance, 
and establishes asymptotic results. 
Section~\ref{sect:experiments} demonstrates our methodological and theoretical developments in a simulation study 
and on empirical data of the StudentLife Study.

\section{Framework and our Network AIPW Estimator}\label{sect:estimator}

\subsection{Model Formulation}\label{sect:model-formulation}

We consider $i=1,\ldots,\NN$ units interacting in a known undirected network $G$. 
For each unit $i$,
we observe a binary treatment $\Wi\in\{0, 1\}$, a univariate outcome $\Yi$, and a possibly multivariate vector of observed covariates $\Ci$ that may causally affect $\Wi$ and $\Yi$.
The outcome $\Yi$ may be dichotomous or continuous, and the potentially multivariate covariates $\Ci$ may consist of discrete and continuous components.
Irrespective of whether the outcomes are continuous or dichotomous, we can consider the following SEM with additive error terms for $i=1,\ldots, N$
\begin{equation}\label{eq:SEM}
	\begin{array}{rcl}
		\Ci &\leftarrow& \epsCi\\
        \Zi & \leftarrow& \big(f^1_z(\Cminusi, G), \ldots, 
	f^t_z(\Cminusi, G)
	\big)\\
		\Wi &\leftarrow&\hzero(\CZi) + \epsWi\\
        \Xi & \leftarrow& \big(f^1_x(\Wminusi, \Cminusi, G), \ldots, 
	f^r_x(\Wminusi, \Cminusi, G)
	\big)\\
		\Yi &\leftarrow& \Wi \gonezero(\CXi) + (1-\Wi)\gzerozero(\CXi) +  \epsYi,
	\end{array}
\end{equation}
where the errors $\epsWi$ and $\epsYi$ are jointly independent, 
we have $\E[\epsWi | \Ci, \Zi] = 0$ and $\E[\epsYi | \Wi, \Ci, \Xi] = 0$, the errors satisfy the assumptions in~\citep{Buehlmann1997} (required for the bootstrap variance results in Appendix~\ref{sect:proofs-thm-var-est-boot}), and the $\epsWi$'s are identically distributed and the $\epsYi$'s are identically distributed (required for the alternative variance results in Appendix~\ref{sect:var_est}).
We note that the identical distribution of the error terms is only required for our approach to estimate standard errors.
The vector $\Cminusi= (C_1, C_2, \ldots, C_{i-1}, C_{i+1}, \ldots, C_N)$ denotes the vector of covariates of units $j\neq i$, and  $\Wminusi$ is similarly defined. The binary treatments $\Wi$ can be thought of as $\Ber(\hzero(\CZi))$ realizations. 
A constant $\hzero$ corresponds to a Bernoulli experiment.
This SEM encodes the assumption that the covariates $\Ci$ and features $\Xi$ suffice to control for confounding of the effect of the treatment on the outcome. The propensity score function $h^0(\cdot, \cdot)$ and the outcome model consisting of $g_1^0(\cdot, \cdot)$ and $g_0^0(\cdot, \cdot)$ are fixed but unknown functions that all units share.
Nevertheless, the distribution of the responses may differ across units due to the $X$-spillover that captures effects from, for example, a unit's neighbors' covariates and treatment assignments as described next. Because every unit may have a different number of neighbors, the $\Xi$'s may follow a different distribution across different units, resulting in non-fixed distributions of the responses across units. Furthermore, the individual equations in~\eqref{eq:SEM} have to be understood in a distributional sense in that, if for example $\gonezero\equiv 0\equiv \gzerozero$, we have $\Yi = \epsYi$ in distribution only.

The functions $f^l_z$, $l\in\indset{t}$ and $f^l_x$, $l\in\indset{r}$, which are shared by all units and used to build the $Z$- and $X$-features, are assumed to be known and their concatenations are assumed to be of fixed dimensions $t$ and $r$, respectively.
This is analogous to 
the in-practice considerations in~\citet{sofrygin2017}. 
We also allow for features of further degree neighbors: for example, $f_x^1$ might capture the fraction of treated units that are a distance of $2$ from a given unit in the network $G$. 
Making use of an implied dependency graph gives a more transparent formulation; see Section~\ref{sec:dependency_graph}. 
Since the network $G$ is undirected, our spillover effects are assumed to be reciprocal; that is, if unit $i$ receives spillover effects from unit $j$ through $\Wj$ and/or $\Cj$, then unit $j$ also receives spillover effects from unit $i$ through $\Wi$ and/or $\Ci$. 
Example~\ref{example:features} illustrates the construction of $2$-dimensional $X$-features. Importantly, the $X$- and $Z$-features render the unit-level data dependent. In addition, the distributions of propensity scores and outcomes are not generally identical across units due to distributional differences of these features.  
 
\begin{figure}
	\centering
	\includegraphics[width=0.3\textwidth]{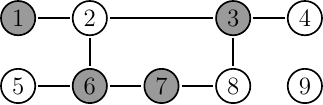}
		\caption[]{\label{fig:network-example} 
	A network on nine units where the node label represents the number of a unit. Gray nodes receive the treatment, corresponding to $\Wi = 1$, and white ones do not, corresponding to $\Wi = 0$. 
	}
	\label{fig:network_features}
\end{figure}

\begin{exmpl}\label{example:features}
Consider the network in Figure~\ref{fig:network_features} where
gray nodes take the treatment and white ones do not. We choose $r=2$ many $X$-features and discard any influence of $C_j$ in $X_i$, that is, $f^l_x\big(\{(\Wj, \Cj)\}_{j\in\indset{\NN}\setminus\{i\}}, G\big) = f^l_x\big(\{(\Wj)\}_{j\in\indset{\NN}\setminus\{i\}}, G\big)$ for $l=1,2.$ Given a unit $i$, we choose the first feature in $\Xi$ as the fraction of treated neighbors of unit $i$ and the second feature as the fraction of treated neighbors of neighbors of $i$. 
Let us consider unit $i=6$ in Figure~\ref{fig:network_features}. Its neighbors are the units $2$, $5$, and $7$, and its neighbors of neighbors are the units $1$ and $3$ (neighbors of unit $2$) and unit $8$ (neighbor of unit $7$), where we exclude $i=6$ from its second degree neighborhood by definition. Therefore, we have $X_6 = (1/3, 2/3)$ because one out of three neighbors is treated and two out of three neighbors of neighbors are treated. The whole $9\times 2$ dimensional $X$-feature matrix is obtained by applying the same computations to all other units $i$.
\end{exmpl}

\subsection{Treatment Effect and Identification}

Plugging in the outcome equation of the SEM \eqref{eq:SEM}, we can rewrite the treatment effect of interest, the EATE, as
\begin{align}\label{eq:thetaN}
	\thetazeroN 
 =& \frac{1}{\NN}\sum_{i=1}^{\NN} \E_{W_{-i}, C_{-i}, C_i}\Big[ \E\left[\Yi^{do(\Wi = 1)} - \Yi^{do(\Wi = 0)}\mid W_{-i}, C_{-i}, C_i\right]\Big]\nonumber\\
 = & \frac{1}{\NN}\sum_{i=1}^{\NN}\E_{\Ci, \Xi}\big[\E_{\epsYi}[\Yi \,\vert\, \text{do}(\Wi = 1),\Ci, \Xi] - \E_{\epsYi}[\Yi \,\vert\, \text{do}(\Wi = 0),\Ci, \Xi]\big]\nonumber\\
 =& \frac{1}{\NN}\sum_{i=1}^{\NN}\E_{\Ci,\Xi}[\gonezero(\CXi) - \gzerozero(\CXi)],
\end{align}
where we get that the unit-specific treatment effect of unit $i$ is $\thetazeroi = \E_{\Ci,\Xi}[\gonezero(\CXi) - \gzerozero(\CXi)]$. Particularly, we assume that given the observable confounders $C_i$ and features $X_i$, we can replace the do-operator by respective conditioning. The expectation $\E_{\Ci,\Xi}$ over $\Ci$ and $\Xi$ is with respect to the observational distributions of $\Ci$ and $\Xi$, as defined by the SEM~\eqref{eq:SEM}. 
This notation makes explicit that the EATE 
is conditional on $\NN$, whereas it remains implicit that it is also conditional on the network $G$. We refer to~\cite{sofrygin2017} for a discussion of the interpretation of such conditional effects. 

Estimating $\gonezero$ and $\gzerozero$ by regression machine learning algorithms and plugging them into~\eqref{eq:thetaN} would not result in a parametric convergence rate and an asymptotic Gaussian distribution of the so-obtained estimator.  
To obtain asymptotic normality with convergence at the $\sqrt{N}$-rate, a centered correction term involving the propensity score $h^0$ is added to $\gonezero(\CXi) - \gzerozero(\CXi)$, and we can identify the EATE as follows. 
\begin{lemma}\label{lem:identifiability}
	Let $i\in\indset{\NN}$. 
    Let 
    \begin{equation} \label{def:data}
    \Si=(\Ci, \Zi, \Wi, \Xi, \Yi)
\end{equation}
be the concatenation of the observed variables for unit $i$. For concatenations $\eta = (\gone, \gzero, \h)$ of general nuisance functions $\gone$, $\gzero$, and $\h$, consider the score 
\begin{equation}\label{eq:score}
	\phi(\Si, \eta) 
	=
	\gone(\CXi) - \gzero(\CXi)
	+ \frac{\Wi}{\h(\CZi)} \big(\Yi - \gone(\CXi)\big)
	- \frac{1 - \Wi}{1 - \h(\CZi)} \big(\Yi - \gzero(\CXi) \big) 
\end{equation}
including the above-mentioned correction term. 
	For the true nuisance functions $\etazero=(\gonezero, \gzerozero, \hzero)$, we have $\E[ \phi(\Si, \etazero) ] = \thetazeroi$ and can consequently identify the EATE \eqref{eq:thetaN} by
	\begin{equation}\label{eq:identifyThetazeroN}
		\thetazeroN = \frac{1}{\NN} \sum_{i=1}^{\NN} \E\big[\phi(\Si, \etazero)\big].
	\end{equation}	
The above expectation is with respect to the law of $\Si$, but we omit it for notational simplicity.
\end{lemma}

The proof of Lemma~\ref{lem:identifiability} is provided in Appendix~\ref{sect:proofs-thm-Gauss}. Based on this lemma, we will present our estimator of $	\thetazeroN$ in Section~\ref{sec:estimation}. 
The true nuisance functions $\etazero=(\gonezero, \gzerozero, \hzero)$ are not of statistical interest, but  have to be estimated to build an estimator of $\thetazeroN$, and we will estimate them 
using regression machine learning algorithms.  Such machine learning estimators might suffer from regularization bias and converge slower than at the $\sqrt{\NN}$-rate. 
However, the two correction terms $\Wi/\h(\CZi) (\Yi - \gone(\CXi))$ and $(1 - \Wi)/(1 - \h(\CZi)) (\Yi - \gzero(\CXi))$ make the score $\phi$ Neyman orthogonal, which counteracts 
the effect of regularization bias. Moreover, the machine learning estimators are only required to converge at a moderate rate; 
 please see Section~\ref{sec:estimation} for further details.

\citet{Scharfstein1999} and~\citet{Robins2005}
consider a similar score $\phi$ for causal effect estimation and inference under the SUTVA assumption, and their function is based on the influence function for the mean for missing data from~\citet{Robins-Rotnitzky1995}. 
Moreover, it is also used to compute the AIPW estimator under SUTVA, 
and our score $\phi$ defined in \eqref{eq:score} coincides with the one of the AIPW approach under SUTVA if we omit the $X$- and $Z$-spillover features. In this case, we can reformulate $\phi$ as 
\begin{displaymath}
    \phi(\Si, \etazero) 
    =
    \frac{\Wi\Yi}{e(\Ci)} - \frac{(1-\Wi)\Yi}{(1-e(\Ci))}
    - \frac{\Wi-e(\Ci)}{e(\Ci)(1 - e(\Ci))} \Big( (1-e(\Ci))\E[\Yi|\Wi = 1, \Ci] + e(\Ci) \E[\Yi|\Wi=0,\Ci] \Big),
\end{displaymath}
where $e(\Ci)  = \E[\Wi|\Ci]= \hzero(\Ci)$ denotes the propensity score,  $\E[\Yi|\Wi = 1, \Ci] = \gonezero(\CXi) $, and $\E[\Yi|\Wi = 0, \Ci] = \gzerozero(\CXi)$. This equivalence remains true if the true nuisance functions are replaced by their estimators.

\subsection{Dependency Graph}\label{sec:dependency_graph}

Depending on the feature functions that are used,
if an edge connects two units in the network $G$, the units may be dependent. However, the absence of an edge in $G$ does not necessarily imply independence of the respective units.
Subsequently, we present a second graph where the presence of an edge represents dependence and its absence independence of the variables of the two respective units. Our theoretical results will be established based on this so-called dependency graph~\citep{Saevje2021}.
Example~\ref{example:dependency_graph} illustrates the concept. 

\begin{figure}
\includegraphics[width=0.038\textwidth]{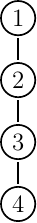} \quad\quad\quad
\includegraphics[width=0.1\textwidth]{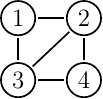}
\quad\quad\quad
\includegraphics[width=0.7\textwidth]{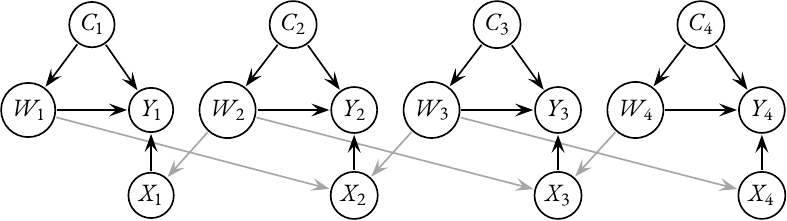}
\caption{\label{fig:network-example2} 
	A network $G$ on four units (left), where the spillover effects come from the treatments of the direct neighbors, 
	which results in a distance-two dependence, which is displayed in the  corresponding dependency graph $\tildeG$ (middle). The underlying causal DAG is displayed on the right, where arrows due to $X$-spillover effects are gray.}
\end{figure}

\begin{defn}Dependency graph on $\Si$, $i\in\indset{\NN}$~\citep{Saevje2021}. 
	The dependency graph $\tildeG=(V, \tildeE)$ on the unit-level data $\Si$, $i\in\indset{\NN}$ defined in \eqref{def:data}, is an undirected graph
	on the node set $V$ of the network $G=(V,E)$ with potentially larger edge set $\tildeE$ than $E$. 
	An undirected edge $\{i,j\}$ between two nodes $i$ and $j$ from $V$ belongs to $\tildeE$ if at least one of the following two conditions holds:
	$1)$ there exists an $m\in\indset{\NN}\setminus\{i,j\}$ such that $\Wm$ and/or $\Cm$ are present in both $\Xi$ and $\Xj$ or are present in both $\Zi$ and $\Zj$; $2)$
	$\Wi$ is present in $\Xj$, or 
	$\Ci$ is present in $\Xj$ or in $\Zj$. 
	That is, units $i$ and $j$ receive spillover effects from at least one common third unit, or they receive spillover effects from each other.
\end{defn}

\begin{exmpl}\label{example:dependency_graph}
Consider the chain-shaped network $G$ in Figure~\ref{fig:network-example2} on the left. 
We consider as $1$-dimensional $X$-spillover effect the fraction of treated direct neighbors in the  network $G$ and no $Z$-spillover. The resulting dependency graph $\tildeG$ is displayed in the middle of Figure~\ref{fig:network-example2}. 
In $\tildeG$, unit $2$ shares an edge with units $1$ and $3$ because these units are neighbors of $2$ in the network. Unit $2$ also shares an edge with $4$ in $\tildeG$ because it shares its neighbor $3$ with unit $4$.
The right of Figure~\ref{fig:network-example2} displays the causal DAG on all units corresponding to this model, including confounders $C$. 
Due to the definition of the $X$-spillover effect, we have $X_1=W_2$ and $X_4=W_3$. Consequently, using graphical criteria~\citep{Lauritzen1996, Pearl1998, Pearl2009, Pearl2010, Perkovic2018}, 
we infer that the unit-level data $S_1= (C_1,W_1,X_1,Y_1)$ is independent of $S_4=(C_4,W_4,X_4,Y_4)$. 
\end{exmpl}

The dependency graph is a function of the network $G$ as well as the $Z$-and $X$-features.
Constraining the growth of the maximal degree of this graph allows us to obtain a CLT result for our treatment effect estimator.

\subsection{Estimation Procedure and Asymptotics} \label{sec:estimation}

Subsequently, we describe our estimation procedure and its asymptotic properties.
We use sample splitting and cross-fitting to estimate the EATE $\thetazeroN$ identified by Equation~\eqref{eq:identifyThetazeroN} as follows. 
We randomly partition $\indset{\NN}$ into $\KK\ge 2$ sets of approximately equal size that we call $I_1,\ldots,I_{\KK}$. 
We split the unit-level data according to this partition into the sets $\SIk = \{\Si\}_{i\in\Ik}$, $\kk\in\indset{\KK}$. For each $\kk\in\indset{\KK}$, we perform the following steps.
First, we estimate the nuisance functions $\gonezero$, $\gzerozero$, and $\hzero$ on the complement set of $\SIk$, which we define as 
\begin{equation}\label{eq:SIkc}
	\SIkc = \{\Sj\}_{j\in\indset{\NN}}\setminus 
	\big( \SIk \cup \{\Sm\ | \ 
	\exists i \in\Ik\colon (i,m)\in\tildeE
	 \} \big), 
\end{equation}
where $\tildeE$ denotes the edge set of the dependency graph $\tildeG$. 
Particularly, $\SIkc$ consists of unit-level data $\Sj$ from units $j$ that do not share an edge with any unit $i\in\Ik$ in the dependency graph. 
Consequently, the set $\SIkc$ contains all $\Sj$'s that are independent of the data in $\SIk$. 
To estimate $\gonezero$, we select the $\Si$'s from $\SIkc$ whose $\Wi$ equals $1$ and regress the corresponding outcomes $\Yi$ on the confounders $C_i$ and the features $X_i$, which yields the estimator $\goneIkc$.
Similarly, to estimate $\gzerozero$, we select the $\Si$'s from $\SIkc$ whose $\Wi$ equal $0$ and perform an analogous regression, which yields the estimator $\gzeroIkc$.
To estimate $\hzero$, we use the whole set $\SIkc$ and regress $\Wi$ on the confounders $C_i$ and the features $Z_i$, which yields the estimator $\hIkc$\footnote{If the treatment is randomized with a known probability, we do not have to estimate the propensity function $\hzero$ and set it to the randomization probability instead.}. 
These regressions may be carried out with any machine learning algorithm.
We concatenate these nuisance function estimators into the nuisance parameter estimator $\hetaIkc = (\goneIkc, \gzeroIkc, \hIkc)$ and plug it into $\phi$ that is defined in~\eqref{eq:score}. We then evaluate the so-obtained function $\phi(\cdot, \hetaIkc)$ on the data $\SIk$, which yields the terms $\phi(\Si, \hetaIkc)$ for $i\in\Ik$. 
That is, we evaluate $\phi(\cdot, \hetaIkc)$ on unit-level data $\Si$ that is independent of the data that was used to estimate the nuisance parameter $\hetaIkc$. 
Finally, we estimate the EATE
by the cross-fitting estimator
\begin{equation}\label{eq:theta-est}
	\htheta = \frac{1}{\KK}\sum_{\kk = 1}^{\KK} \Bigg(\frac{1}{\normone{\Ik}}\sum_{i\in\Ik}\phi(\Si, \hetaIkc)\Bigg)
\end{equation}
that averages over all $\KK$ folds. The estimator $\htheta$ converges at the parametric rate, $\NN^{-1/2}$, and follows a Gaussian distribution asymptotically with limiting variance $\sigma_{\infty}^2$ as stated in Theorem \ref{thm:Gaussian} below.

The partition $I_1, \ldots, I_{\KK}$ is random. To alleviate the effect of this randomness, the whole procedure is repeated a number of $\Salg$ times, and the median of the individual point estimators over the $\Salg$ repetitions is our final estimator of $\thetazeroN$. The asymptotic results for this median estimator remain the same as for $\htheta$; see~\citet{Chernozhukov2018}. 
For each repetition $\salg \in [\Salg]$, we compute a point estimator $\hthetas$, a variance estimator $\hvars$ (for details please see the next Section \ref{sect:boot-var}), and a p-value $\ps$ for the two-sided test $H_0\colon \thetazeroN = 0$ versus $H_A\colon \thetazeroN\neq 0$. 
The $\Salg$ many p-values $p_1, \ldots, p_{\Salg}$ from the individual repetitions are aggregated according to
\begin{displaymath}\label{eq:medAggregate}
	\paggrzero = 2 \median_{\salg\in\indset{\Salg}}(\ps). 
\end{displaymath}
This aggregation scheme yields a valid overall p-value for the same two-sided test~\citep{Meinshausen2012}. 
The corresponding confidence interval is constructed as
\begin{equation}\label{eq:CI}
    \CI(\htheta) = 
	\{\theta\in\R\ | \ \paggrTheta \mathrm{\ of \ testing \ } H_0\colon\thetazeroN = \theta \mathrm{\ vs. \ } H_A\colon \thetazeroN\neq\theta
	\mathrm{\ satisfies \ } \paggrTheta > \alpha\}, 
\end{equation}
where typically $\alpha = 0.05$. 
This set contains all values $\theta$ 
for which the null hypothesis $H_0\colon\thetazeroN = \theta$ cannot be rejected at level $\alpha$ against the two-sided alternative $H_A\colon \thetazeroN\neq\theta$. 

Next, we describe how $\CI(\htheta)$ can easily be computed. 
Due to the asymptotic result of Theorem~\ref{thm:Gaussian}, 
the aggregated p-value $\paggrTheta$ for $\theta\in\R$ can be represented as 
\begin{displaymath}
	\paggrTheta = 4 \median_{\salg\in\indset{\Salg}} \big( 1 - \Phi(\sqrt{\NN}\hsigmas^{-1} \normone{\hthetas-\theta}) \big),
\end{displaymath}
where $\Phi$ denotes the cumulative distribution function of a standard Gaussian random variable. 
Consequently, we have
\begin{displaymath}
	\paggrTheta > \alpha
	\quad\Leftrightarrow\quad
	\Phi^{-1}(1 - \alpha/4) > \median_{\salg\in\indset{\Salg}} (\sqrt{\NN}\hsigmas^{-1} \normone{\hthetas-\theta}), 
\end{displaymath}
which can be solved for feasible values of $\theta$ using root search. 
A full description of our method is presented in Algorithm~\ref{algo:Summary}. 

\begin{algorithm}[h!] 
	\SetKwInOut{Input}{Input}
   	\SetKwInOut{Output}{Output}

 \Input{$\NN$ unit-level observations $\Si=(\Wi, \CXZi, \Yi)$ from the model~\eqref{eq:SEM}, 
  network $G$, feature functions $f^l_x$, $l\in\indset{r}$ and $f^l_z$, $l\in\indset{t}$, 
 corresponding dependency graph $\tildeG$, 
  natural number $\KK$, natural number $\Salg$, significance level $\alpha\in [0, 1]$, machine learning algorithms. 
 }
 \Output{Estimator of the EATE $\thetazeroN$ and a valid p-value and confidence interval for the two-sided test $H_0\colon\thetazeroN=0$ vs. $H_A\colon\thetazeroN\neq 0$.
 }
 
 \For{$\salg\in\indset{\Salg}$}
 {
 Randomly split the index set $\indset{\NN}$ into $\KK$ sets $I_1, \ldots, I_{\KK}$ of approximately equal size.
 
 \For{$\kk\in\indset{\KK}$}
 {
 Compute nuisance function estimators $\goneIkc$, $\gzeroIkc$, and $\hIkc$ with  
 machine learning algorithm and data from $\SIkc$. 
 }

Compute point estimator of $\thetazeroN$ according to~\eqref{eq:theta-est}, and call it $\hthetas$. 

Estimate asymptotic variance of $\hthetas$ 
using the bootstrap procedure described in Section~\ref{sect:boot-var}
 (or according to Theorem~\ref{thm:var-est} in Appendix~\ref{sect:var_est}), and call it $\hvars$.

Compute p-value $\ps$ for the two-sided test $H_0\colon\thetazeroN=0$ vs. $H_A\colon\thetazeroN\neq 0$ using $\hthetas$, $\hvars$, and asymptotic Gaussian approximation.
 }
 
Compute $\htheta = \median_{s\in\indset{\Salg}}(\hthetas)$.
 
 Compute aggregated p-value $\paggrzero = 2 \median_{\salg\in\indset{\Salg}}\ps$. 
 
 Compute confidence interval according to~\eqref{eq:CI}, call it $\CI(\htheta)$.

Return $\htheta$, $\paggrzero$, $\CI(\htheta)$.

 \caption{Estimating the EATE from observational data on networks with spillover effects using plugin machine learning}\label{algo:Summary}
\end{algorithm}

Before we present our main theorem we mentioned in the construction of confidence intervals above, we present and discuss key assumptions. 
First, we require that products of machine learning errors decay fast enough, namely 
\begin{displaymath}
			\begin{array}{l}
			\normP{\hzero(\CZi)-\hIkc(\CZi)}{2}
			\cdot
			\bigg(\normP{\gonezero(\CXi)-\goneIkc(\CXi)}{2}\\
			\quad\quad\quad
			+ 
			\normP{\gzerozero(\CXi)-\gzeroIkc(\CXi)}{2} + 
			\normP{\hzero(\CZi)-\hIkc(\CZi)}{2} 
			\bigg)\ll \NN^{-\frac{1}{2}}; 
			\end{array}
		\end{displaymath}
see Assumption~\ref{assumpt:DML} in the appendix for more details. 
In particular, the individual error terms may vanish at a  rate 
smaller than $\NN^{-1/4}$. This is achieved by many machine learning  
methods under suitable assumptions; see for instance  
\citet{Chernozhukov2018}:
$\ell_1$-penalized and related methods in a variety of sparse models
\citep{Bickel2009, Buehlmann2011, Belloni2011, Belloni-Chernozhukov2011, Belloni2012, Belloni-Chernozhukov2013}, forward selection in sparse models
\citep{Kozbur2020}, $L_2$-boosting in sparse linear models
\citep{Luo2016}, a class of regression trees and random forests 
\citep{Wager2016}, and neural networks \citep{Chen1999}. 
Second, to ensure enough sparsity in the dependency structure of the data, the maximal degree $\dmax$ 
in the dependency graph is assumed to grow at most
at the rate $\dmax = o(\NN^{1/4})$, which implies that the dependencies are not too far reaching. This assumption allows us to bound the Wasserstein-distance of our (centered and scaled) treatment effect estimator to a standard Gaussian random variable
using Stein's method~\citep{stein}. 
\begin{restatable}{assumptions}{assumptdegree}\label{assumpt:degree} 
The maximal degree $\dmax$ of a node in the dependency graph satisfies $\dmax = o(\NN^{1/4})$. 
\end{restatable} 
\citet{sofrygin2017} only require $\dmax = o(\NN^{1/2})$, but 
achieve a slower convergence rate of their treatment effect estimator. To recover the $\sqrt{\NN}$-rate, they require that $\dmax$ is bounded by a constant, meaning $\dmax = O(1)$. 

Furthermore, we require that this dependency structure is not too strong moment-wise in the sense that the variance term given in the following assumption converges. 
\begin{restatable}{assumptions}{assumptvariance}\label{assumpt:variance}
	Let $\{\PcalN\}_{\NN\ge 1}$ be a sequence of sets of probability distributions $\PP$ of the $\NN$ units.  There exists $\sigmainfty^2$, possibly depending on $\PP \in \PcalN$, satisfying 
 $0 < L \le \sigmainfty^2 \le U < \infty$ with fixed constants $L, U$, such that
	for all $\PP\in\PcalN$, we have
	\begin{equation}\label{eq:product-property}
		\lim_{\NN\rightarrow\infty}\bigg(\Var\bigg(\frac{1}{\sqrt{\NN}}\sum_{i=1}^{\NN}\psi(\Si, \thetazeroi, \etazero)\bigg) - \sigmainfty ^2 \bigg) = 0,
	\end{equation}
    where $\psi(\Si, \thetazeroi, \etazero) = \phi(\Si, \etazero) - \thetazeroi$ is a centered version of $\phi$.
\end{restatable}

Assuming bounded second moments, $\sum_{j=1}^{\NN}\normone{\Cov(\psi(\Si, \thetazeroi, \etazero),\psi(\Sj, \thetazeroi, \etazero))}$ can be bounded, up to constants, by $\degr{i}$, where $\degr{i}$ denotes the degree of node $i$ in the dependency graph. Consequently, we have
\begin{equation}\label{eq:need-to-O1}
    \Var\bigg(\frac{1}{\sqrt{\NN}}\sum_{i=1}^{\NN}\psi(\Si, \thetazeroi, \etazero)\bigg)
    \le \gamma \cdot \frac{1}{\NN}\sum_{i=1}^{\NN} \degr{i},
\end{equation}
where $\gamma$ denotes some universal constant. 
Subsequently, we consider two special cases. First, if the maximal degree of the dependency graph is uniformly bounded by some constant $D$, we can bound~\eqref{eq:need-to-O1} by the constant $\gamma D$. 
Second, assume the dependency graph has some nodes with finite degree: $\degr{i}\le D$ for $i$ in some set $S_{\mathrm{max}}^c$; the other nodes' degree $\degr{i}$ for $i\in S_{\mathrm{max}}$ is bounded by $\dmax=o(\NN^{1/4})$ with $\normone{S_{\mathrm{max}}} \ge O(\NN/ \dmax) = O(\NN^{3/4})$. 
Then, \eqref{eq:need-to-O1} is also of bounded order $O(1)$. 

\begin{theorem}[Asymptotic distribution of $\htheta$]\label{thm:Gaussian}
	Assume Assumption~\ref{assumpt:degree} and \ref{assumpt:variance} as well as~\ref{assumpt:regularity} and~\ref{assumpt:DML} stated in the appendix in Section~\ref{sect:AssumptionsDefinitions}. 
	Then, the estimator $\htheta$ of the EATE $\thetazeroN$ given in \eqref{eq:theta-est} converges at the parametric rate, $\NN^{-1/2}$, and asymptotically follows a Gaussian distribution, namely
	\begin{equation}\label{eq:Gauss}
		\sqrt{\NN}\sigmainfty^{-1}(\htheta-\thetazeroN) \stackrel{d}{\rightarrow} \mathcal{N}(0, 1) \quad (\NN\to\infty), 
	\end{equation}
	where $\sigmainfty$ is characterized in Assumption~\ref{assumpt:variance}. 
	The convergence in~\eqref{eq:Gauss} is in fact uniformly over the law $\PP \in \PcalN$ $(\NN\to\infty)$.
\end{theorem}
Please see Section~\ref{sect:proofs-thm-Gauss} in the appendix for a proof of Theorem~\ref{thm:Gaussian}. 
The asymptotic variance $\sigmainfty^2$ in Theorem~\ref{thm:Gaussian} can be consistently estimated using a bootstrap approach; 
see Section~\ref{sect:boot-var}. Alternatively, it is possible to consistently estimate it using a plugin approach; see Theorem~\ref{thm:var-est} in the next Section \ref{sect:var_est}. 
However, empirical simulations have revealed that the bootstrap procedure described in the next section performs better. 

Our estimator $\htheta$ is robust in two senses. First, it is $\sqrt{\NN}$-consistent and asymptotically normal if only 
the product property~\eqref{eq:product-property} of the machine learning estimators holds. Second, it can be shown that it remains consistent if either the propensity model or the outcome model are correctly specified. These properties are also called rate double robustness and model double robustness, respectively~\citep{Rotnitzky2019}.

\subsection{Bootstrap Variance Estimator}\label{sect:boot-var}

We use the residual bootstrap as follows to estimate the asymptotic variance. 
First, we use the estimated nuisance functions to compute the outcome regression residuals. More precisely, for $i\in\indset{\NN}$, denote by $\kk(i)$ the index in $\indset{\KK}$
specifying the partition unit $i$ belongs to, namely $i\in\Ikdegr{i}$. Then, we estimate the $\epsY$'s by $\hepsYi = \hepsYi' - \frac{1}{\NN}\sum_{j=1}^{\NN}\hepsYj'$, where 
        $\hepsYi' = \Yi - \Wi\goneIkci{i}(\Ci, \Xi)  - (1-\Wi)\gzeroIkci{i}(\Ci, \Xi) $.
Next, we sample confounders $\{\Ci^*\}_{i\in\indset{\NN}}$  
with replacement from $\{\Ci\}_{i\in\indset{\NN}}$, and we sample $\hepsYi^*$  
with replacement from $\{\hepsYi\}_{i\in\indset{\NN}}$. These sampled covariates and error terms are now propagated through the SEM~\eqref{eq:SEM}, that is, we compute 
        $\Zi^* = (f^1_z(\Cminusi^*, G), \ldots, 
	f^t_z(\Cminusi^*, G)
	)$, sample $\Wi^* = \mathrm{Bernoulli}(\hIkci{i}(\Ci^*,\Zi^*))$, compute $\Xi^*= (f^1_x(\Wminusi^*, \Cminusi^*, G), \ldots, 
	f^r_x(\Wminusi^*, \Cminusi^*, G)
	)$, and build $\Yi^* = \Wi^*\goneIkci{i}(\Ci^*, \Xi^*)  - (1-\Wi^*)\gzeroIkci{i}(\Ci^*, \Xi^*) + \hepsYi^*$.  Subsequently, we concatenate these values to obtain the bootstrap datapoints $\Si^*=(\Ci^*, \Zi^*, \Wi^*, \Xi^*, \Yi^*)$, $i\in\indset{\NN}$. 
	Then, we apply our treatment effect estimation procedure to the $\Si^*$'s to obtain a bootstrap estimator $\hat{\theta}^{*}$. 
 This procedure is repeated $R$ many times, and the bootstrap variance estimator is given by the empirical variance of the $\hat{\theta}^{*}_r$ over $r\in\indset{R}$. 

\begin{restatable}{theorem}{boostrap-var}\label{thm:var-est-boot}
    The bootstrap scheme described in Section~\ref{sect:boot-var} consistently estimates the asymptotic variance~\eqref{eq:product-property} under Assumption~\ref{assumpt:var-boot} stated in the appendix. 
\end{restatable}

The proof of Theorem~\ref{thm:var-est-boot} can be found in Appendix~\ref{sect:proofs-thm-var-est-boot}.

\section{Empirical Validation}\label{sect:experiments}

We demonstrate our method in a simulation study and on a real-world dataset. In the simulation study, we validate the performance of our method on different network structures and compare it to 
two popular treatment effect estimators.
Afterwards, we investigate the effect of study time on  exam performance in the Swiss StudentLife Study~\citep{Stadtfeld2019, Voeroes2021} taking into account the effect of social ties.

\subsection{Simulation Study}\label{sect:simulation}

We investigate a fairly simple data generating mechanism with $1$-dimensional $X$-features and no $Z$-features. The $X$-interference effects a unit receives come from an interaction between treatments and control of its immediate neighbors in the network (we consider Erd{\H{o}}s--R{\'e}nyi and Watts--Strogatz).  
We compare the performance of our method to two popular off-the-shelf alternative schemes with respect to bias of the point estimator and coverage and length of respective two-sided confidence intervals: the H\'ajek estimator and an IPW estimator. Our aim is to see that these standard estimators may suffer in the presence of interference and to demonstrate that our easy-to-implement estimator overcomes their shortcomings.

We first describe the two competitors and afterwards detail the simulation setting and present the results.
Our code is available on GitHub (\url{https://github.com/corinne-rahel/networkAIPW}).

The \textbf{H\'ajek estimator} (denoted by ``Hajek'' in Figure \ref{fig:simulation}) without incorporation of confounders~\citep{Hajek1971} equals 
\begin{displaymath}
    \frac{1}{\NN}\sum_{i=1}^{\NN}\bigg( \frac{\Wi\Yi}{\frac{1}{\NN}\sum_{j=1}^{\NN}\Wi}
    + \frac{(1-\Wi)\Yi}{\frac{1}{\NN}\sum_{j=1}^{\NN}(1-\Wi)}\bigg).
\end{displaymath}
The parametric convergence rate and asymptotic Gaussian distribution are preserved under $X$-spillover effects that equal the fraction of treated neighbors in a randomized experiment~\citep{Li-Wager2022}. The \textbf{IPW estimator}~\citep{Rosenbaum1987} has been developed under SUTVA and uses observed confounding by creating a ``pseudo population'' in which the treatment is independent of the confounders~\citep{Hirano2003}.
We compute it using sample splitting and cross-fitting according to 
\begin{displaymath}
    \frac{1}{\KK} \sum_{\kk=1}^{\KK}
    \frac{1}{\normone{\Ik}}\sum_{i\in\Ik}
    \Big(\frac{\Wi\Yi}{\eIkc(\Ci)}
    - \frac{(1-\Wi)\Yi}{1-\eIkc(\Ci)}\Big),
\end{displaymath}
where $\eIkc$ is the fitted propensity score obtained by regressing $\Wi$ on $\Ci$ on the data in $i\in\SIkc$. 
In our simulation, $\eIkc$ coincides with $\hIkc$ because we consider no $Z$-features. We denote this estimator by ``IPW'' in Figure~\ref{fig:simulation}. These estimators are not designed for the interference structures we consider, but we would like to investigate the performance of these off-the-shelf and easy to implement estimators, also in comparison to our proposed method.

\begin{figure}
\centering
\includegraphics[height=0.3\textwidth]{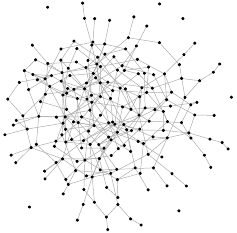}
~\quad\quad\quad~
\includegraphics[height=0.3\textwidth]{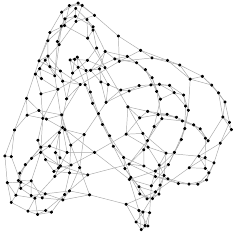}

\caption{\label{fig:graphs}
	Different network structures on $\NN = 200$ units: Erd{\H{o}}s--R{\'e}nyi network (left) where two nodes are connected with probability $3/\NN$ (every node is connected to $3$ other nodes in expectation); Watts--Strogatz network (right) with a rewiring probability of $0.05$, a $1$-dimensional ring-shaped starting lattice where each node is connected to $2$ neighbors on both sides (that is, every node is connected to 4 other nodes),
	no loops, and no multiple edges. The graphs are generated using the \textsf{R}-package \texttt{igraph}~\citep{igraph}.}
\end{figure}

\begin{figure}[h!] 
	\centering 
	\includegraphics[width=\textwidth]{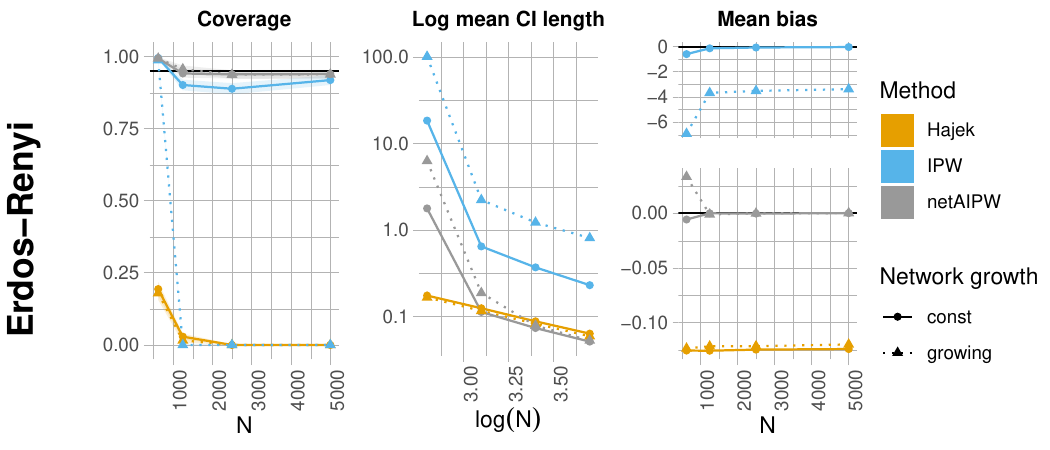}
	\includegraphics[width=\textwidth]{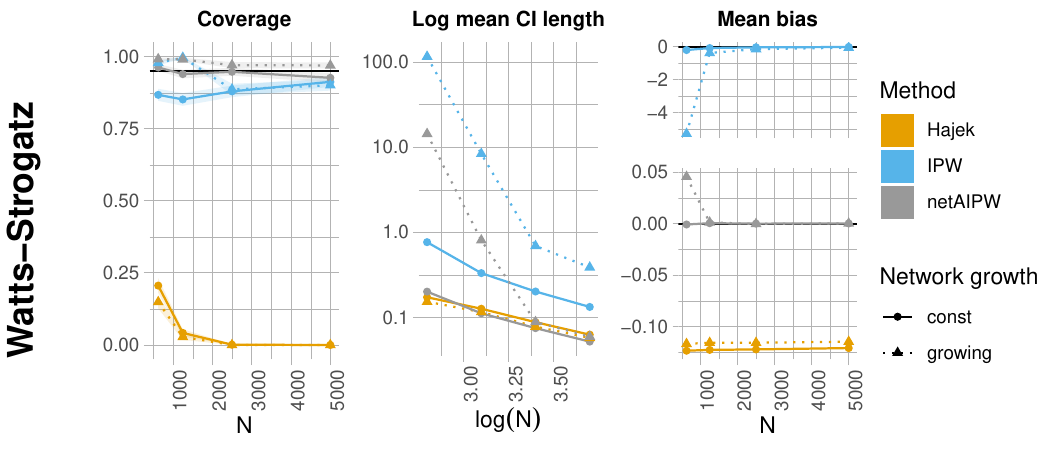}
		\caption[]{\label{fig:simulation} 
	Coverage (fraction of times the true, and in general unknown, $\thetazeroN$ was inside the confidence interval), log mean length of two-sided $95\%$ confidence intervals for $\thetazeroN$, 
	and mean bias over $1000$ simulation runs for Erd\H{o}s--R\'enyi and Watts--Strogatz networks of different complexities (Erd\H{o}s--R\'enyi: expected degree $3$ and $3\NN^{1/15}$ for ``const'' and ``$\NN {\mathchar"5E} (1 / 15)$'', respectively; Watts--Strogatz: before rewiring, nodes have degree $4$ and $4\NN^{1/15}$ for ``const'' and ``$\NN {\mathchar"5E} (1 / 15)$'', respectively, and the rewiring probability is $0.05$). 
	We compare the performance of our method, netAIPW, with the  H\'ajek  
	and an  IPW 
	estimator, indicated by color. The variance of the competitors are empirical variances over the $1000$ repetitions, whereas we computed confidence intervals for netAIPW according to~\eqref{eq:CI} with $B=1$ and $300$ bootstrap samples. The shaded regions in the coverage plot represent $95\%$ confidence bands with respect to the $1000$ simulation runs. 
	}
\end{figure}

We investigate two network structures that govern our interference effects: Erd{\H{o}}s--R{\'e}nyi networks~\citep{Erdoes1959} and   Watts--Strogatz networks~\citep{watts1998}. 
Erd{\H{o}}s--R{\'e}nyi networks randomly form edges between units with a fixed probability and are a simple example of a random mathematical network model.  
These networks play an important role as a standard against which to compare more complicated models. 
Watts--Strogatz networks, also called small-world networks, share two properties with many networks in the real world: a small average shortest path length and a large clustering coefficient.
To construct such a network, the vertices are first arranged in a regular fashion and linked to a fixed number of their neighbors. 
Then, some randomly chosen edges are rewired with a constant rewiring probability. A representative of each network type is provided in Figure~\ref{fig:graphs}. For each of these two network types, we consider one case where the dependency in the network does not increase with $\NN$ (denoted by ``const'' in Figure \ref{fig:simulation}) and one where it increases with $\NN$ (denoted by $\NN {\mathchar"5E} (1 / 15)$ in Figure \ref{fig:simulation}).

The specific unit-level structural equations~\eqref{eq:SEM} we consider are as follows.
For each unit $i\in\indset{\NN}$, we sample \iid\ confounders $\Ci\sim\mathrm{Unif}(0, 1)$ from the uniform distribution. 
The treatment selections $\Wi$ are drawn from a Bernoulli distribution with arbitrarily chosen success probability $p_i = p_i(\Ci) = 0.15 \one_{\Ci < 0.33} + 0.5 \one_{0.33 \le \Ci < 0.66} + 0.85 \one_{0.66 \le \Ci}$. Let $\alpha(i)$ denote the neighbors of unit $i$ in the network (without $i$ itself). Then, we let the $1$-dimensional $X$-features $\Xi$ denote the shifted average number of neighbors assigned to treatment weighted  by their confounder, namely 
\begin{displaymath}
    \Xi = \frac{1}{\normone{\alpha(i)}}\sum_{j\in\alpha(i)} (\one_{\Wj = 1} - \one_{\Wj=0} ) \Cj,
\end{displaymath}
if $\alpha(i)$ is non-empty, and $0$ else. We do not consider $Z$-features.
For real numbers $x$ and $c$, we consider the arbitrary functions
\begin{displaymath}
     \gonezero(x, c) =
     1.5 \one_{x \ge 0.5,c \ge -0.2,x < 0.7} 
     + 4 \one_{c \ge -0.2,x \ge 0.7} 
     + 0.5 \one_{x \ge 0.5, c < -0.2} 
     + 3.5 \one_{x < 0.5, c \ge -0.2} 
     + 2.5 \one_{x < 0.5, c < -0.2}
\end{displaymath}
and 
\begin{displaymath}
        \gzerozero(x, c) = 
        0.5 \one_{x \ge 0.4, c \ge 0.2} 
     - 0.75 \one_{x \ge 0.4, c < 0.2} 
     + 0.25 \one_{x < 0.4, c \ge 0.2} 
     - 0.5 \one_{x < 0.4, c < 0.2}. 
\end{displaymath}
That is, the functions $\gonezero$, $\gzerozero$, and $\hzero$ are step functions.  
For \iid\ error terms $\epsYi\sim\mathrm{Unif}(-\sqrt{0.12} / 2, \sqrt{0.12} / 2)$, we consider the outcomes $\Yi = \Wi\gonezero(\Ci, \Xi) + (1-\Wi)\gzerozero(\Ci,\Xi) + \epsYi$. 

For the sample sizes $N= 625, 1250, 2500, 5000$, we perform $1000$ simulation runs redrawing the data according to the SEM, consider $\Salg = 1$, $\KK = 5$, and $R=300$ bootstrap samples to estimate the variance in Algorithm~\ref{algo:Summary}.
That is, we consider one split per generated dataset and consequently do not aggregate $p$-values in these simulations. However, the empirical analysis in Section~\ref{sect:empirical} aggregates $p$-values over $100$ datasplits.
 We estimate the nuisance functions by random forests   
consisting of $500$ trees with a minimal node size of $5$ and other default parameters using the \textsf{R}-package \texttt{ranger}~\citep{ranger}. 
To estimate the propensity score, we limit the depth of the trees to $2$. Our results for the Erd\H{o}s--R\'enyi and Watts--Strogatz networks are displayed in Figure~\ref{fig:simulation}.
Two different panels are used to display the results for different ranges of the bias of the methods. 
For all network types and complexities, we observe the following. 
The IPW estimator 
 incurs some bias as can be expected because it  
does not account for network spillover and even under SUTVA, it is not Neyman orthogonal, which means we are not allowed to plug in machine learning estimators of  nuisance functions. Furthermore, it is known to have a poor finite-sample performance 
due to estimated propensity scores $\eIkc$ that may be close to $0$ or $1$. 
The H\'ajek estimator incurs some bias because it does not adjust for observed confounding and assumes a randomized treatment instead.
The bias of our method (denoted by ``netAIPW'' in Figure~\ref{fig:simulation}) decreases as the sample size increases. As the dependency graph becomes more complex, our method requires more observations to achieve a small bias because the data sets $\SIkc$ in~\eqref{eq:SIkc}, which are used to estimate the nuisance functions, are smaller in denser networks. 
In terms of coverage, the two competitors perform poorly,  whereas our method guarantees coverage. 

Simulation results involving spillover effects from second degree neighbors and misspecified spillover effects are presented in Appendix~\ref{sect:second-degree}.
Furthermore, for a $\mathrm{Bernoulli}(1/2)$ treatment assignment and with the ``const'' Watts-Strogatz setting presented in the main paper, we found that the AIPW approach leads to variances that are of about a factor of $23$ smaller than the ones obtained with IPW.
This suggests that AIPW is helpful in reducing the variance of IPW even in the randomized case.

\subsection{Empirical Analysis: Swiss StudentLife Study Data}\label{sect:empirical}

Subsequently, we estimate the causal effect of study time on academic success of university students with our newly developed estimator. We quantify this causal effect by the EATE that is the average of the difference in expected grade point average (GPA) of the final exam had a student studied much versus little, allowing for potential spillover effects from the student's friends on the student's study time.
Among the factors that determine academic success are person-specific traits, such as intelligence~\citep{chamorro-premuzic_personality_2008},
willingness to work hard~\citep{los_interaction_2019},  
and socioeconomic background~\citep{Heckman2006}. 
The Swiss StudentLife Study data~\citep{Stadtfeld2019, Voeroes2021} was collected to investigate the impact of various factors
on academic achievement.  
It consists of observations from freshmen undergraduate students pursuing a degree in the natural sciences at a Swiss university. Instead of a university entrance test, these students had to pass a demanding examination after one year of studying. At several time points throughout this year, the students were asked to fill out questionnaires about their student life, social network, and well-being. 
The data consists of three cohorts of students. 
Cohort~$1$ was observed in 2016 and cohorts~$2$ and~$3$ in 2017. Importantly, for all three cohorts, the data contains friendship information among the students. We build the corresponding undirected network by drawing an edge between two students if at least one of them mentioned the other one as being a friend. We believe that spillover effects arise due to students interacting in this network, and thus we have to control for them when estimating the EATE described above. 
Figure~\ref{fig:SSL-graphs} displays the resulting network consisting of the three cohorts.

GPA ($\Yi$) 
constitutes our outcome variable and represents the average grade of seven to nine exams, depending on study programs. 
It ranges from 1 to 6, with passing grades of 4 or higher.
The average GPA in the data we used was $4.266$ with a standard deviation of $0.872$. 
The remaining variables were measured five to six months before the exam period and correspond to wave four of the Swiss StudentLife Study data. 
The self-reported number of hours spent studying per week during the semester ($\Wi$) 
constitutes the treatment variable. It was dichotomized into studying many ($\Wi=1$) and few ($\Wi=0$) hours. We considered a setting where $\Wi=1$ corresponds to studying at least $8$ hours per week, which is the $20\%$ quantile, and one where $\Wi=1$ corresponds to studying at least $20$ hours per week, which is the $80\%$ quantile. 
We consider spillover effects from the friends of a student, which are a student's direct neighbors in the friendship network. 
We consider $Z$-spillover effects that account for the effect of befriended students' study motivation and stress variables on a student's treatment. We do not consider spillover effects on the outcome GPA (no $X$-features). 
The $\Zi$-spillover variable of a student $i$ is a vector of length $6$, where each entry corresponds to the average of the following six variables across the friends of the student:  
(a) study motivation, 
measured with the learning objectives subscale of the SELLMO-ST\footnote{\label{foot:sellmo}This is a scale to assess learning and achievement motivation, and the subscale consists of eight items measured on a five-point Likert-scale from 1 (``completely disagree'') to 5 (``completely agree'').}~\citep{Spinath2002}, 
(b) work avoidance, 
measured with the work avoidance subscale of the students version of the SELLMO-ST\footref{foot:sellmo},
(c) the average of ten perceived stress items~\citep{Cohen1988},
(d, e) two items specifically on exam related stress, 
and (f) whether one was perceived as clever by at least one other student. 
In addition to these network effects, we control on the unit level ($\Ci$) for the just mentioned variables observed on an individual unit as well as the cohort number, gender, 
having Swiss nationality, 
speaking German, 
and the financial situation. 
From all the data of the three cohorts combined, we only considered individuals for whom all the mentioned variables, that is, treatment, outcome, covariates, and $Z$-spillover variables, are observed. The final sample consisted of $\NN= 526$ individuals: $113$ from cohort~1, $119$ from cohort~2, and $294$ from cohort~3. 
In our algorithm, we used $\S=1000$ sample splits (from which we aggregate p-values as in~\eqref{eq:medAggregate}) with $K = 10$ groups each and random forests consisting of $5000$ trees to learn $\gzerozero$, $\gonezero$, and $\hzero$ whose leaf size  
was initially determined by $5$-fold cross-validation. 
Also, we used the variance estimator as in Appendix~\ref{sect:var_est} that relies on fewer assumptions.

\begin{figure}
\centering
\includegraphics[height=0.25\textwidth]{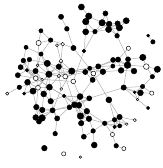}
~\quad~
\includegraphics[height=0.25\textwidth]{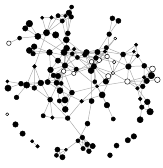}
~\quad~
\includegraphics[height=0.25\textwidth]{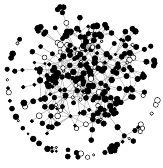}
\caption{\label{fig:SSL-graphs} 
    Friendship networks per cohort
    with black dots representing $\Wi=1$ and a weekly study time of at least $8$ hours, white for $\Wi=0$ and a weekly study time of less than $8$ hours, and a bigger node size represents a higher GPA.
    }
\end{figure}

We estimated the EATE with two different definitions for $\Wi = 1$, defined by a study time of either at least $8$ or $20$ hours per week, corresponding to the $20\%$ and $80\%$ quantiles, respectively, and Table~\ref{tab:SSL-results} displays the results. 
Table~\ref{tab:SSL-results-low} displays our estimated EATE with $\Wi=1$ representing a weekly study time of at least $8$ hours. 
Our EATE estimator is positive and significant. On average, students received a $0.362$ points higher GPA 
had they studied at least $8$ hours per week compared to studying less.
Consequently, a significantly higher GPA can be achieved by studying more. 
If we apply the same procedure but exclude the $Z$-spillover covariates (no spillover), the EATE estimator is  higher and also significant. 
Table~\ref{tab:SSL-results-high} displays our results with $\Wi=1$ representing a weekly study time of at least $20$ hours.
Our EATE estimator is positive but not significant anymore. 
Hence, our results suggest that GPA is not significantly higher had a student studied at least $20$ hours per week compared to studying less.
Without spillover, the treatment effect is significant. In both cases in Table~\ref{tab:SSL-results}, the estimate of the EATE is higher under the assumption of no spillover effects, compared to the estimator that allows for possible $Z$-spillover effects. This potentially relevant difference highlights the importance of not a priori ruling out spillover effects.
Overall, the model including spillover effects seems more realistic than the one excluding them. 
Finally, when interpreting the results, it is important to recall that study time captures the learning time during the semester. 
There is an additional eight-week lecture-free preparation period, 
and our study time does not reflect this preparation time. 
Consequently, our results only describe the EATE of study time during the semester on GPA.

\begin{table}
    \begin{subtable}[h]{0.5\textwidth}
        \centering
        \begin{tabular}{l c c}
        \textbf{Spillover} & \textbf{EATE} & \textbf{$95\%$ CI for $\thetazeroN$} \\
        \hline
        yes & $0.362$ & $[0.283, 0.442]$\\
        no & $0.451$ & $[0.364, 0.528]$\\
       \end{tabular}
       \caption{$\Wi=1$ if studied at least $8$ hours per week ($20\%$ quantile).}
       \label{tab:SSL-results-low}
    \end{subtable}
    \hfill
    \begin{subtable}[h]{0.5\textwidth}
        \centering
        \begin{tabular}{l c c}
        \textbf{Spillover} & \textbf{EATE} & \textbf{$95\%$ CI for $\thetazeroN$} \\
        \hline
        yes & $0.078$ & $[-0.096,  0.252]$\\
        no & $0.163$ & $[0.011, 0.311]$\\
       \end{tabular}
       \caption{$\Wi=1$ if studied at least $20$ hours per week ($80\%$ quantile).}
       \label{tab:SSL-results-high}
    \end{subtable}
     \caption{\label{tab:SSL-results}EATE and $95\%$ confidence intervals for $\thetazeroN$ for different settings with different control 
     groups, namely studying less than 8 (a) or less than 20 (b) hours per week. 
     }
\end{table}

\section{Conclusion}\label{sect:conclusion}

Causal inference with observational data usually assumes independent units. 
However, having independent observations is often questionable, and so-called spillover effects among units are common in practice. 
Our aim was to develop point estimation and asymptotic inference for the expected average treatment effect (EATE) with observational data from a single (social) network. We would like to point out the hardness of this problem: 
we consider treatment effect estimation on data with increasing dependence among units, where the data generating mechanism can be highly nonlinear and include confounders. 
We use an augmented inverse probability weighting (AIPW) principle
and account for spillover effects 
that we capture by features, which are functions of the known network and the treatment and covariate vectors.   
There may be several features, and one feature may capture spillover effects from different units than another feature; these units might be direct neighbors to compute one feature and neighbors of neighbors to compute another feature. We consider the dependency graph to pose assumptions on these features in our asymptotic theory.
Units may interact beyond their direct neighborhoods, interactions may become increasingly complex as the sample size increases, and we consider arbitrary networks. 
Using ideas of double machine learning~\citep{Chernozhukov2018}, we develop a cross-fitting algorithm under interference that allows us
to estimate the nuisance components of our model by arbitrary machine learning algorithms. 
Although we employ machine learning algorithms, our EATE estimator converges at the $\sqrt{\NN}$-rate and asymptotically follows a Gaussian distribution, which allows us to perform inference. 

In a simulation study, we demonstrated that commonly employed methods for treatment effect estimation suffer from the presence of spillover effects, whereas our method could account for the complex dependence structures in the data so that the bias vanished with increasing sample size and coverage was guaranteed. In the Swiss StudentLife Study, 
we investigated the EATE of study time on the grade point average of university examinations,
accounting for spillover effects due to friendship relations.  Omitting this spillover
may lead to biased results due to spurious association. 

In the present work, we focused on estimating the EATE. Other effects may be estimated in a similar manner, like for instance the global average treatment effect (GATE) where all units are jointly intervened on. 
We develop an estimator of the GATE in Appendix~\ref{sect:extension}.

\section*{Acknowledgements}

CE and PB received
funding from the European Research Council (ERC) under the European Union’s Horizon 2020 research and innovation programm (grant agreement No. 786461), and M-LS received funding from the Swiss National Science Foundation (SNF) (project No. 200021\_172485). 
The Swiss StudentLife data collection was supported by Swiss National Science Foundation Grant 10001A 169965 and the rectorate of ETH Zurich.
We also
thank Leonard Henckel and Dominik Rothenh\"ausler for useful
comments.

\bibliography{references}

\begin{appendices}

\section{Assumptions and Additional Definitions}\label{sect:AssumptionsDefinitions}

We consider the following notation. 
We denote by $\indset{\NN}$ the set $\{1,2,\ldots,\NN\}$. We add the probability law as a subscript to the probability operator $\Prob$ and the expectation operator $\E$ whenever we want to emphasize the corresponding dependence.
We denote the $L^p(\PP)$-norm by $\normP{\cdot}{p}$ and the Euclidean or operator
norm by $\normone{\cdot}$, depending on the context. 
We implicitly assume that given expectations and conditional expectations exist. We denote by $\stackrel{d}{\rightarrow}$ convergence in distribution.
The symbol $\independent$ denotes independence of random variables.

We observe $\NN$ units according to the structural equations~\eqref{eq:SEM} that are connected by an underlying network. 
For each unit $i\in\indset{\NN}$, we concatenate $\Si = (\Wi, \CXZi, \Yi)$ that are relevant for unit $i$. 
For notational simplicity, we abbreviate $\Di = (\CXi)$ and $\Ui = (\CZi)$ for $i\in\indset{\NN}$.

Let the number of sample splits $\KK\ge 2$ be a fixed integer independent of $\NN$. We assume that $\NN\ge\KK$ holds. Consider a partition $I_1, \ldots,I_{\KK}$ of $\indset{\NN}$. 
We assume that all sets $I_1, \ldots,I_{\KK}$ are of equal cardinality $\nn$. 
We make this assumption for the sake of notational simplicity, but our results hold without it. 

Let $\{\deltaN\}_{\NN\ge \KK}$ and $\{\DeltaN\}_{\NN\ge \KK}$ be two sequences of non-negative numbers that converge to $0$ as $\NN\rightarrow\infty$. 
Let $\{\PcalN\}_{\NN\ge 1}$ be a sequence of sets of probability distributions $\PP$ of the $\NN$ units. 

For completeness, we recall the following two assumptions from the main text. Assumption~\ref{assumpt:degree} that limits the growth rate of the maximal degree of a node in the dependency graph. Assumption~\ref{assumpt:variance} characterizes the asymptotic variance in Theorem~\ref{thm:var-est} as the limit of the population variance on the $\NN$ units. 

\assumptdegree*

\assumptvariance*

We make the following additional sets of assumptions. 
The following Assumption~\ref{assumpt:regularity} recalls that we use the model~\eqref{eq:SEM} and specifies regularity assumptions on the involved random variables. Assumption~\ref{assumpt:regularity2} and~\ref{assumpt:regularity3} ensure that the random variables are integrable enough.
Assumption~\ref{assumpt:regularity4} ensures that the true underlying function $\hzero$ of the treatment selection model is bounded away from $0$ and $1$, which allows us to divide by $\hzero$ and $1-\hzero$. 

\begin{assumptions}\label{assumpt:regularity} 
Let $p\ge 4$. 
	For all $\NN$, all $i\in\indset{\NN}$,  all $\PP\in\PcalN$, and all $\kk\in\indset{\KK}$, we have the following. 
	\begin{enumerate}[label={\theassumptions.\arabic*}]
		\item\label{assumpt:regularity1}
		The structural equations~\eqref{eq:SEM} hold, where the treatment $\Wi\in\{0, 1\}$ is binary. 
	
		\item\label{assumpt:regularity2}
		There is a finite real constant $\Cnormp$ independent of $\PP$ satisfying $\normP{\Wi}{p}+ \normP{\Ci}{p}+ \normP{\Xi}{p} + \normP{\Zi}{p} + 
		\normP{\Yi}{p}\le \Cnormp$.
		
		\item\label{assumpt:regularity3}
		        There is a finite real constant $\Cnorminfty$ independent of $\PP$ such that we  have $\normP{\Yi}{\infty} + \normP{\gonezero(\Di)}{\infty} + \normP{\gzerozero(\Di)}{\infty} + \normP{\hzero(\Ui)}{\infty} \le \Cnorminfty$. 
		        
		\item\label{assumpt:regularity4}
		There is a finite real constant $\Chzero$ independent of $\PP$ such that $\PP(\Chzero\le \hzero(\Ui)\le 1-\Chzero) = 1$ holds. 
		
		\item\label{assumpt:regularity5}
		There is a finite real constant $\Ctheta$ such that we have $\normone{\thetazeroi}\le \Ctheta$. 
	
	\end{enumerate}
\end{assumptions}

The following Assumption~\ref{assumpt:DML} characterizes the realization set of the nuisance functions and the $\NN^{-1/2}$ convergence rate of products of the machine learning errors from estimating the nuisance functions $\gonezero$, $\gzerozero$, and $\hzero$.

\begin{assumptions}\label{assumpt:DML}
	Consider the $p\ge 4$ from Assumption~\ref{assumpt:regularity}.  
	For all $\NN\ge \KK$ and all $\PP\in\PcalN$, 
	consider a nuisance function realization set $\TauN$ such that the following conditions hold.
	\begin{enumerate}[label={\theassumptions.\arabic*}]
		\item\label{assumpt:DML1}
		The set $\TauN$ consists of $\PP$-integrable functions $\eta=(\gone, \gzero, \h)$ whose $p$th moment exists and whose $\normP{\cdot}{\infty}$-norm is in fact uniformly bounded, and $\TauN$ contains $\etazero = (\gonezero, \gzerozero, \hzero)$.  Furthermore, there is a finite real constant $\CnormEta$ such that for all $i\in\indset{\NN}$ and all elements $\eta = (\gzero, \gone, \h)\in\TauN$, we have
		\begin{displaymath}
			\normP{\hzero(\Wi)-\h(\Wi)}{2}
			\cdot
			\big(\normP{\gonezero(\Di)-\gone(\Di)}{2}+ 
			\normP{\gzerozero(\Di)-\gzero(\Di)}{2} + 
			\normP{\hzero(\Wi)-\h(\Wi)}{2} 
			\big)
			\le\deltaN\NN^{-\frac{1}{2}}.
		\end{displaymath}

		\item\label{assumpt:DML4}
		Assumption~\ref{assumpt:regularity4} also holds with $\hzero$ replaced by $\h$. 
		
		\item\label{assumpt:DML3}
		Let $\kappa$ be the largest real number such that for all $i\in\indset{\NN}$ and all $\eta\in\TauN$, we have
		\begin{displaymath}
			\normP{\hzero(\Wi)-\h(\Wi)}{2}
			+ \normP{\gonezero(\Di)-\gone(\Di)}{2}
			+ \normP{\gzerozero(\Di)-\gzero(\Di)}{2}
			\lesssim \sqrt{\deltaN}\NN^{-\kappa}.
		\end{displaymath}
		That is, $\kappa$ represents the slowest convergence rate of our machine learners. 
		Then, there is a finite real constant $\CproductAndDegre$  such that $\dmax \NN^{-2\kappa}\le \CproductAndDegre$ holds, where $\dmax$ denotes the maximal degree of the dependency graph. 
		
		\item\label{assumpt:DML2}
		For all $\kk\in\indset{\KK}$,  the nuisance parameter estimate $\hetaIkc=\hetaIkc(\SIkc)$ belongs to the nuisance function realization set $\TauN$ 
		with $\PP$-probability no less than $1-\DeltaN$.
	\end{enumerate}
\end{assumptions}

The following two assumptions, Assumption~\ref{assumpt:Acald} and~\ref{assumpt:kappa}, are only required to establish that our plugin estimator of the asymptotic variance is consistent in Appendix~\ref{sect:var_est}. (However, please recall that we recommend using the bootstrap procedure presented in Section~\ref{sect:boot-var} unless the sample size is large). They are not required to establish the asymptotic Gaussian distribution of our plugin machine learning estimator. 

Assumption~\ref{assumpt:Acald} characterizes the order of the minimal size of the sets $\Acald$ for $d\ge 0$. 
These sets are required to contain a sufficient number of units such that the degree-specific treatment effects $\thetazerod$ for $d\ge 0$
can be estimated at a fast enough rate. These estimators are required to give a consistent estimator of the 
asymptotic variance $\sigmainfty^2$.

\begin{assumptions}\label{assumpt:Acald}
	For $d\ge 0$, 
	the order of $\normone{\Acald}$ is at least $\NN^{3/4}$, denoted by $\Omega(\NN^{3/4})$ according to the Bachmann--Landau notation~\citep{Lattimore2020}. 
\end{assumptions}

Assumption~\ref{assumpt:kappa} specifies that all individual machine learning estimators of the nuisance functions converge at a rate faster than $\NN^{-1/4}$.

\begin{assumptions}\label{assumpt:kappa}
	The slowest convergence rate $\kappa$ in Assumption~\ref{assumpt:DML3} satisfies $\kappa\ge 1/4$. 
\end{assumptions}

\section{Network Effects in the Social Sciences} \label{sect:app_socsci}

We consider models related to  spillover effects. However, another notion of spillover effects has prevailed within the social science networks literature, namely social influence effects.  In this appendix, we describe social influence effects  and how their modeling differs from our approach.
Whereas spillover effects represent new covariates on the unit-level that are built from variables of other units along network paths, social influence effects mostly concern effects that a specific variable $A_j$ of neighboring units has on $A_i$ of the $i$th unit. 
In the statistics literature, this is called contagion~\citep{ugander2013, EcklesKarrerUgander+2017}. 
In the social sciences, there are two important models to investigate social influence / contagion processes: the autologistic actor attribute model (ALAAM;~\citet{Robins2001, daraganova_robins_2012}) and the stochastic actor-oriented model (SAOM;~\citet{snijders_2005, SNIJDERS201044, Steglich2010}). 
Both models aim at estimating the degree to which a variable $A_i$ of a focal individual is associated with the values of its neighbors' values of $A$. Whereas ALAAMs only considers cross-sectional data, SAOMs additionally allow estimating longitudinal social influence effects.

In contrast, the spillover features that we consider summarize variables from neighboring units. 
They represent a new variable that is used for the treatment or outcome regression models.
For example, in our empirical analysis, we consider the spillover effect of study motivation of unit $i$'s neighbors on the learning hours of unit $i$. We do not consider spillover from the learning hours of unit $i$'s neighbors on unit $i$'s own learning hours (i.e. social influence / contagion). Instead, we model such associations of the individual units' learning hours by constructing features from other variables and units that act as observed confounders. 
Moreover, we are not interested in estimating the effect as such of, say, other units' study motivation on the learning hours of unit $i$. However, this is possible with ALAAMs and SAOMs. 
We are not interested in estimating spillover as such, but we consider spillover as a tool to control for spurious associations due to the network structure to estimate treatment effects.

\section{Additional Simulation Results }\label{sect:second-degree}

First, we present simulation results involving spillover effects from second degree neighbors and misspecified spillover effects. We consider the same data generating mechanism and estimation framework as in Section~\ref{sect:simulation} apart from the following change: the ``neighborhood'' $\alpha(i)$ defining $\Xi$ contains all second-degree neighbors of unit $i$, that is, all units that are a distance $2$ away from unit $i$ in the network (neighbors of neighbors). 
For incorrectly specified spillover effects, we assumed that $\alpha(i)$ contains the direct neighbors of unit $i$ instead. 
We consider an Erd\H{o}s--R\'enyi network as in Section~\ref{sect:simulation} except that the average degree of a unit is now $2.5$. 
The results are displayed in Figure~\ref{fig:simulation-neighb2}. 
Our method, netAIPW, does not seem to suffer much from the misspecified spillover effects in terms of coverage whereas the other methods do. 
In general, we observed that it is advantageous to include spillover effects even if they are not entirely correctly specified. 

\begin{figure}[h!] 
	\centering 
	\includegraphics[width=\textwidth]{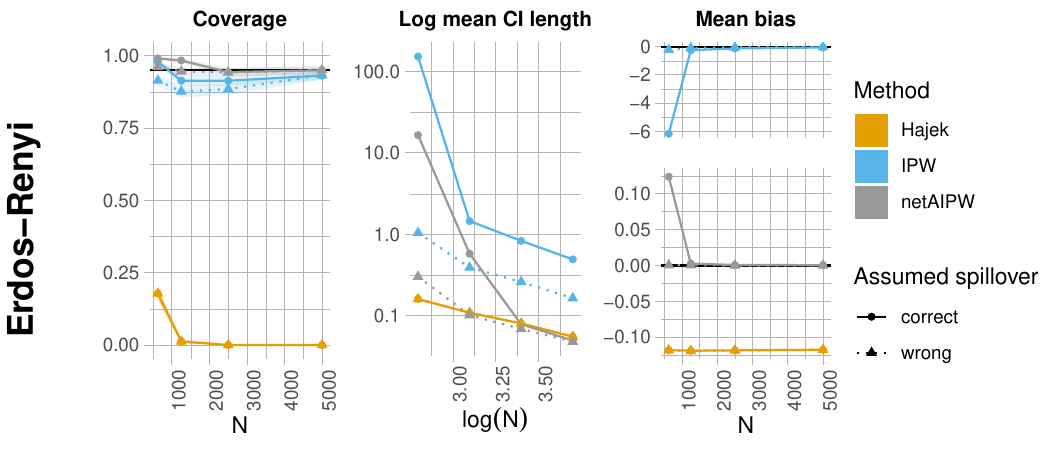}
		\caption[]{\label{fig:simulation-neighb2} 
	Coverage (fraction of times the true, and in general unknown, $\thetazeroN$ was inside the confidence interval), log mean length of two-sided $95\%$ confidence intervals for $\thetazeroN$, 
	and mean bias over $1000$ simulation runs for the ``const'' Erd\H{o}s--R\'enyi network as in Section~\ref{sect:simulation}, except for the average degree of $2.5$. 
	We compare the performance of our method, netAIPW, with the  H\'ajek  
	and an  IPW 
	estimator, indicated by color, for correctly and incorrectly specifying the spillover effects from second-degree neighbors. The variance of the competitors are empirical variances over the $1000$ repetitions, whereas we computed confidence intervals for netAIPW according to~\eqref{eq:CI} with $B=1$ and $300$ bootstrap samples. The shaded regions in the coverage plot represent $95\%$ confidence bands with respect to the $1000$ simulation runs. 
	}
\end{figure}

Next, we present networks and different kinds of spillover effects to show when Assumption~\ref{assumpt:variance} holds and when it fails to hold. We consider the same second-degree spillover effects and Erd\H{o}s--R\'enyi network as above in this section. The ``const'' network
has an expected degree of $2.5$, and the ``$\NN {\mathchar"5E} (1 / 15)$'' one has an expected degree of $2.5\NN^{1/15}$ in Figure~\ref{fig:check-degree}. The maximal degree, divided by $N^{1/4}$ of the dependency graph of the ``const'' network decreases with $\NN$, whereas the respective quantity increases with $\NN$ for the non-constant-degree network. That is, only the constant-degree network satisfies Assumption~\ref{assumpt:variance}. The non-constant-degree network implies a dependency graph that is ``too dense'' to satisfy this assumption. 
We would like to remark that satisfying Assumption~\ref{assumpt:variance} is an interplay of the underlying network and the chosen spillover effects because they determine the dependency graph, and hence its maximal degree, together. A given network might lead to a dependency graph satisfying Assumption~\ref{assumpt:variance} with one kind of spillover effects (e.g., only from neighbors), whereas the same network might lead to a dependency graph violating this assumption with another kind of spillover effects (e.g., also including second-degree neighbors, that is, neighbors of neighbors). 

\begin{figure}[h!] 
	\centering 
	\includegraphics[width=0.7\textwidth]{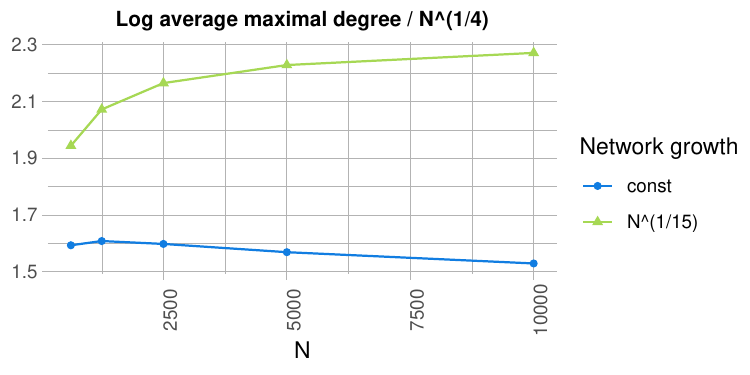}
		\caption[]{\label{fig:check-degree} 
  Simulated maximal degree of the dependency graph from second-degree spillover on Erd\H{o}s--R\'enyi networks with an expected degree of either $2.5$ (``const'') or $2.5\NN^{1/15}$ (``$\NN {\mathchar"5E} (1 / 15)$'') divided by $N^{1/4}$, averaged over $1000$ simulation runs. 
  }
\end{figure}

\section{Supplementary Lemmata}\label{sect:suppl-lem}

In this section, we prove two results on conditional independence relationships of the variables from our model. 
We argue for the directed acyclic graph (DAG) of our model~\eqref{eq:SEM} and use graphical criteria~\citep{Lauritzen1996, Pearl1998, Pearl2009, Pearl2010, Peters2017, Perkovic2018, Maathuis2019}. 
We denote the direct causes of $\Wi$ by $\pa(\Wi)$, the parents of $\Wi$. Analogously, we denote the parents of $\Yi$ by $\pa(\Yi)$; please see for instance~\citet{Lauritzen1996}. 
We assume that $\pa(\Wi)$ consists of $\Ci$ and the variables used to compute the spillover feature $\Zi$ and that $\pa(\Yi)$ consists of $\Wi$, $\Ci$, and the variables used to compute the spillover feature $\Xi$.

\begin{lemma}\label{lem:lemma1}
Let $i\in\indset{\NN}$, and let $\Cj\not\in\pa(\Yi)$. Then, we have $\Yi\independent\Cj | \pa(\Yi)$. 
\end{lemma}
\begin{proof}[Proof of Lemma~\ref{lem:lemma1}]
    The parents of $\Yi$ are a valid adjustment set~\citep{Pearl2009}. 
    Because $\Yi$ has no descendants, the claim follows. 
\end{proof}

\begin{lemma}\label{lem:lemma2}
Let $i\in\indset{\NN}$, and let $\Cj\not\in\pa(\Wi)$. Then, we have $\Wi\independent\Cj | \pa(\Wi)$. Furthermore, for $j\neq i$, we have $\Wi\independent\Wj|\pa(\Wi)$. 
\end{lemma}
\begin{proof}[Proof of Lemma~\ref{lem:lemma2}]
    The parents of $\Wi$ are a valid adjustment set~\citep{Pearl2009}.
    The treatment variable $\Wi$ has no descendants apart from outcomes $Y$, which are colliders on any path from $\Wi$ to $\Cj$ or $\Wj$, and thus the empty set blocks these paths.
    Consequently, the two claims follow. 
\end{proof}

\section{Proof of Theorem~\ref{thm:Gaussian}}\label{sect:proofs-thm-Gauss}

\begin{proof}[Proof of Lemma~\ref{lem:identifiability}]
Let $i\in\indset{\NN}$. We have
\begin{displaymath}
		\E[\psi(\Si, \thetazeroi, \etazero)]\\
		= \E\bigg[ \frac{\Wi}{\hzero(\Ui)} \big(\Yi-\gonezero(\Di)\big) \bigg] -  
		\E\bigg[ \frac{1-\Wi}{1-\hzero(\Ui)} \big(\Yi-\gzerozero(\Di)\big) \bigg]. 
\end{displaymath}
We have
\begin{equation}\label{eq:applicationLemma1}
    \begin{array}{cl}
    & \E\Big[ \frac{\Wi}{\hzero(\Ui)} \big(\Yi-\gonezero(\Di)\big) \Big] \\
    =& \E\Big[ \frac{\Wi}{\hzero(\Ui)} \big(\E[\Yi|\pa(\Yi)\cup\pa(\Wi)]-\gonezero(\Di)\big) \Big] \\
    =& \E\Big[ \frac{1}{\hzero(\Ui)} \E[\Wi\Yi - \Wi\gonezero(\Di)|\pa(\Yi)] \Big] \\
    =& \E\Big[ \frac{\Wi}{\hzero(\Ui)} \E[\epsYi|\pa(\Yi)] \Big] \\
    =& 0
    \end{array}
\end{equation}
due to Lemma~\ref{lem:lemma1} and because $\E[\epsYi|\pa(\Yi)]=0$ holds by assumption. Analogous computations for $\E[ (1-\Wi)/(1-\hzero(\Ui)) (\Yi-\gzerozero(\Di)) ]$ conclude the proof. 
\end{proof}

The following lemma shows that the score function $\phi$ is Neyman orthogonal in the sense that its Gateaux derivative vanishes~\citep{Chernozhukov2018}. 

\begin{lemma}[Neyman orthogonality]\label{lem:Neyman-orth}
Assume the assumptions of Theorem~\ref{thm:Gaussian} hold. 
Let $\eta\in\TauN$, and let $i\in\indset{\NN}$. Then, we have 
\begin{displaymath}
	\frac{\partial}{\partial r}\Big\vert_{r = 0}\E\big[\phi\big(\Si, \etazero + r(\eta - \etazero)\big)\big] = 0. 
\end{displaymath}
\end{lemma}
\begin{proof}[Proof of Lemma~\ref{lem:Neyman-orth}]
	Let $r\in (0,1)$, let $i\in\indset{\NN}$, and let $\eta\in\TauN$. 
	Then, we have
\begin{equation}\label{first-derivative}
	\begin{array}{cl}
		&\frac{\partial}{\partial r}\E\big[\phi\big(\Si, \etazero + r(\eta - \etazero)\big)\big] \\
		=& \frac{\partial}{\partial r} 
		\E\Big[ \gonezero(\Di) - \gzerozero(\Di) + r\big(\gone(\Di)- \gzero(\Di) - \gonezero(\Di)+\gzerozero(\Di)\big)\\
		&\quad + \frac{\Wi}{\hzero(\Ui) + r\big(\h(\Ui)-\hzero(\Ui)\big)} \Big(\Yi - \gonezero(\Di) - r\big(\gone(\Di)-\gonezero(\Di)\big)\Big)\\
		&\quad - \frac{1 -\Wi}{1 - \hzero(\Ui) - r\big(\h(\Ui)-\hzero(\Ui)\big)} \Big(\Yi - \gzerozero(\Di) - r\big(\gzero(\Di)-\gzerozero(\Di)\big)\Big)\Big]\\
		=& \E\Big[
		\big(\gone(\Di)- \gzero(\Di)\big) - \big(\gonezero(\Di)-\gzerozero(\Di)\big)\\
		&\quad + \frac{\Wi}{\big(\hzero(\Ui) + r\big(\h(\Ui)-\hzero(\Ui)\big)\big)^2}\Big(-\big( \gone(\Di)-\gonezero(\Di) \big) \big(\hzero(\Ui) + r(\h(\Ui)-\hzero(\Ui))\big) \\
		&\quad\quad
		- \big(\Yi - \gonezero(\Di) - r(\gone(\Di) - \gonezero(\Di))\big) \big(\h(\Ui) - \hzero(\Ui)\big)\Big)\\
		&\quad - 
		\frac{1 - \Wi}{\big(1 - \hzero(\Ui) - r\big(\h(\Ui)-\hzero(\Ui)\big)\big)^2}\Big(-\big( \gzero(\Di)-\gzerozero(\Di) \big) \big(1 - \hzero(\Ui) - r(\h(\Ui)-\hzero(\Ui))\big) \\
		&\quad\quad
		+ \big(\Yi - \gzerozero(\Di) - r(\gzero(\Di) - \gzerozero(\Di))\big) \big(\h(\Ui) - \hzero(\Ui)\big)\Big)
		\Big].
	\end{array}
\end{equation}
We evaluate this expression at $r=0$ and obtain
\begin{displaymath}
	\begin{array}{cl}
		&\frac{\partial}{\partial r}\Big\vert_{r = 0}\E\big[\phi\big(\Si, \etazero + r(\eta - \etazero)\big)\big] \\
		=&  \E\Big[
		\big(\gone(\Di)- \gzero(\Di)\big) - \big(\gonezero(\Di)-\gzerozero(\Di)\big)\\
		&\quad
		- \Big(1 + \frac{\epsWi}{\hzero(\Ui)}\Big)\big(\gone(\Di) - \gonezero(\Di)\big)
		- \frac{\Wi}{(\hzero(\Ui))^2}\big( \Yi-\gonezero(\Di) \big) \big(\h(\Ui)-\hzero(\Ui)\big)\\
		&\quad
		+ \Big(1 - \frac{\epsWi}{1 - \hzero(\Ui)}\Big)\big(\gzero(\Di) - \gzerozero(\Di)\big)
		- \frac{1 -\Wi}{(1-\hzero(\Ui))^2}\big( \Yi-\gzerozero(\Di) \big) \big(\h(\Ui)-\hzero(\Ui)\big)\Big]\\
		=&0
	\end{array}
\end{displaymath}
due to~\eqref{eq:applicationLemma1} and because
\begin{displaymath}
    \begin{array}{cl}
        &\E\Big[
         \frac{\epsWi}{\hzero(\Ui)}
         \big(\gone(\Di) - \gonezero(\Di)\big)\Big]\\
         =& \E\Big[
         \big(\E[\Wi|\pa(\Wi)\cup\pa(\Yi)] - \hzero(\Ui)\big)
         \frac{1}{\hzero(\Ui)}
         \big(\gone(\Di) - \gonezero(\Di)\big)
        \Big]\\
        =&\E\Big[
         \E[\Wi - \hzero(\Ui)|\pa(\Wi)]
         \frac{1}{\hzero(\Ui)}
         \big(\gone(\Di) - \gonezero(\Di)\big)
        \Big]\\
        =&\E\Big[
         \E[\epsWi|\pa(\Wi)]
         \frac{1}{\hzero(\Ui)}
         \big(\gone(\Di) - \gonezero(\Di)\big)
        \Big]\\
        =& 0
    \end{array}
\end{displaymath}
holds due to Lemma~\ref{lem:lemma2} and because we assumed $\E[\epsWi|\pa(\Wi)]$ = 0, and similarly for $\E[
         \epsWi/(1 - \hzero(\Ui))
         (\gzero(\Di) - \gzerozero(\Di))]$. 

\end{proof}

The following lemma bounds the second directional derivative of the score function. Its proof uses that products of the errors of the machine learners are of a smaller order than $\NN^{-1/2}$. 

\begin{lemma}[Product property]\label{lem:product-property}
Assume the assumptions of Theorem~\ref{thm:Gaussian} hold. 
Let $r\in(0, 1)$, let $\eta\in\TauN$, and let $i\in\indset{\NN}$. Then, we have
\begin{displaymath}
	\normonebigg{\frac{\partial^2}{\partial r^2}\E\big[\phi\big(\Si, \etazero + r(\eta - \etazero)\big)\big]}\lesssim \deltaN\NN^{-\frac{1}{2}}. 
\end{displaymath}
\end{lemma}
\begin{proof}[Proof of Lemma~\ref{lem:product-property}]
We use the first directional derivative we derived in~\eqref{first-derivative} to compute the second directional derivative
\begin{displaymath}
	\begin{array}{cl}
		& \frac{\partial^2}{\partial r^2}\E\big[\phi\big(\Si, \etazero + r(\eta - \etazero)\big)\big]\\
		=& 
		2\E\Big[
		\frac{\Wi}{\big( \hzero(\Ui) + r(\h(\Ui)-\hzero(\Ui)) \big)^4}
		\Big(\big( \gone(\Di)-\gonezero(\Di) \big) \big(\hzero(\Ui) + r(\h(\Ui)-\hzero(\Ui))\big) \\
		&\quad\quad\quad
			+ \big(\Yi - \gonezero(\Di)- r(\gone(\Di) - \gonezero(\Di))\big) \big(\h(\Ui) - \hzero(\Ui)\big) \Big) \\
		&\quad\quad
		\cdot \Big( \hzero(\Ui) + r\big(\h(\Ui) - \hzero(\Ui)\big) \Big) 
		\big( \h(\Ui) - \hzero(\Ui) \big)
		\Big]\\
		& \quad
		+ 2\E\Big[
		\frac{1-\Wi}{\big( 1 - \hzero(\Ui) - r(\h(\Ui)-\hzero(\Ui)) \big)^4}
		\Big(\big( \gzero(\Di)-\gzerozero(\Di) \big) \big(1 - \hzero(\Ui) - r(\h(\Ui)-\hzero(\Ui))\big) \\
		&\quad\quad\quad
			- \big(\Yi - \gzerozero(\Di)- r(\gzero(\Di) - \gzerozero(\Di))\big) \big(\h(\Ui) - \hzero(\Ui)\big) \Big) \\
		&\quad\quad
		\cdot \Big( 1 - \hzero(\Ui) - r\big(\h(\Ui) - \hzero(\Ui)\big) \Big) 
		\big( \h(\Ui) - \hzero(\Ui) \big)
		\Big]. 
	\end{array}
\end{displaymath}
Due to H{\"o}lder's inequality and Assumption~\ref{assumpt:regularity1}, \ref{assumpt:regularity3}, \ref{assumpt:regularity4}, and \ref{assumpt:DML1}, we have
\begin{displaymath}
	\begin{array}{cl}
		&\normonebigg{\frac{\partial^2}{\partial r^2}\E\big[\phi\big(\Si, \etazero + r(\eta - \etazero)\big)\big]}\\
		\lesssim & 
		\big(\normP{\gone(\Di) - \gonezero(\Di)}{2} + \normP{\h(\Ui) - \hzero(\Ui)}{2}\big)\normP{\h(\Ui) - \hzero(\Ui)}{2}\\
		&\quad 
		+ \big(\normP{\gzero(\Di) - \gzerozero(\Di)}{2} + \normP{\h(\Ui) - \hzero(\Ui)}{2}\big)\normP{\h(\Ui) - \hzero(\Ui)}{2}. 
	\end{array}
\end{displaymath}
Due to Assumption~\ref{assumpt:DML1}, both summands above are bounded by $\deltaN\NN^{-1/2}$, and hence we conclude the proof. 
\end{proof}

The following lemma describes how we apply Stein's method~\citep{chin2021} to obtain the asymptotic Gaussian distribution of our estimator although the data is highly dependent. 

\begin{lemma}[Asymptotic distribution with Stein's method]\label{lem:Stein}
Assume the assumptions of Theorem~\ref{thm:Gaussian} hold. 
	Denote by 
	\begin{displaymath}
		\sigmaN^2 = \Var\bigg(\frac{1}{\sqrt{\NN}}\sum_{i=1}^{\NN}\psi(\Si, \thetazeroi, \etazero)\bigg).
	\end{displaymath}
	Observe that by Assumption~\ref{assumpt:variance}, we have $\lim_{\NN\to\infty}(\sigmaN^2 -\sigmainfty^2)= 0$.
	Then, we have
	\begin{displaymath}
		\sigmaN ^{-1}\cdot\frac{1}{\sqrt{\NN}} \sum_{i=1}^{\NN} \psi(\Si, \thetazeroi, \etazero) \stackrel{d}{\rightarrow} \mathcal{N}(0, 1). 
	\end{displaymath}
\end{lemma}
\begin{proof}[Proof of Lemma~\ref{lem:Stein}]
According to Lemma~\ref{lem:identifiability}, we have $\E[\psi(\Si, \thetazeroi, \etazero)] = 0$. 
According to Assumption~\ref{assumpt:regularity}, the fourth moment of $\psi(\Si, \thetazeroi, \etazero)$ exists for all $i\in\indset{\NN}$ and is uniformly bounded over $i\in\indset{\NN}$. 
Recall that we denote by $\dmax$ the maximal degree in the dependency graph on $\Si$, $i\in\indset{\NN}$. 
Due to~\citet[Theorem 3.6]{Ross2011}, we can thus bound the Wasserstein distance of $\sigmaN ^{-1}\cdot\frac{1}{\sqrt{\NN}} \sum_{i=1}^{\NN} \psi(\Si, \thetazeroi, \etazero)$ to $\mathcal{N}(0, 1)$ as follows: there exist finite real constants $c_1$ and $c_2$ such that we have
\begin{equation}\label{eq:wasserstein}
	\begin{array}{cl}
		& d_{\mathcal{W}}\bigg(\sigmaN ^{-1}\cdot\frac{1}{\sqrt{\NN}} \sum_{i=1}^{\NN} \psi(\Si, \thetazeroi, \etazero)\bigg)\\
		\le & 
		c_1\cdot\frac{\dmax^{3/2}}{\sigmaN^2}\cdot\sqrt{\sum_{i=1}^{\NN}\E\Big[\big(\frac{1}{\sqrt{\NN}}\psi(\Si, \thetazeroi, \etazero)\big)^4\Big]}
		+ c_2 \cdot\frac{\dmax^2}{\sigmaN^3}\cdot \sum_{i=1}^{\NN} \E\Big[\normonebig{\frac{1}{\sqrt{\NN}}\psi(\Si, \thetazeroi, \etazero)}^3\Big]\\
		=& c_1\cdot\frac{\dmax^{3/2}\cdot \frac{1}{\sqrt{\NN}}}{\sigmaN^2}\cdot\sqrt{\frac{1}{\NN}\sum_{i=1}^{\NN}\E[\psi^4(\Si, \thetazeroi, \etazero)]}
		+ c_2 \cdot\frac{\dmax^2\cdot \frac{1}{\sqrt{\NN}}}{\sigmaN^{3}}\cdot \frac{1}{\NN}\sum_{i=1}^{\NN} \E\big[\normone{\psi(\Si, \thetazeroi, \etazero)}^3\big].
	\end{array}
\end{equation}
By assumption, we have $\dmax = o(\NN^{1/4})$. Thus, we have $\dmax^{3/2}\cdot \frac{1}{\sqrt{\NN}} = o(\NN^{-1/8})$ and $\dmax^2\cdot \frac{1}{\sqrt{\NN}} = o(1)$. 
Because the terms $\E[\psi^4(\Si, \thetazeroi, \etazero)]$ and $\E\big[\normone{\psi(\Si, \thetazeroi, \etazero)}^3\big]$ are uniformly bounded over all $i\in\indset{\NN}$
and because $\sigmaN\to\sigmainfty$ as $\NN\to\infty$ according to Assumption~\ref{assumpt:variance}, the Wasserstein distance in~\eqref{eq:wasserstein} is of order $o(1)$. Consequently, we infer the statement of the lemma. 
\end{proof}

\begin{lemma}[Vanishing covariance due to sparse dependency graph]\label{lem:vanishing-cov}
Assume the assumptions of Theorem~\ref{thm:Gaussian} hold. 
Let $k\in\indset{\KK}$, and recall that $\nn=\normone{\Ik}$ holds. Then, we have
\begin{displaymath}
	\normonebigg{
	\frac{1}{\sqrt{\nn}} \sum_{i\in\Ik}\big( \phi(\Si, \hetaIkc) - \E[ \phi(\Si, \hetaIkc)|\SIkc] \big) - 
	\frac{1}{\sqrt{\nn}} \sum_{i\in\Ik}\big( \phi(\Si, \etazero) - \E[ \phi(\Si, \etazero)] \big)
	}
	= o_{\PP}(1).
\end{displaymath}
\end{lemma}
\begin{proof}[Proof of Lemma~\ref{lem:vanishing-cov}]
Let $k\in\indset{\KK}$.
We have
\begin{equation}\label{eq:to-bound-vanishing-cov-lemma}
	\begin{array}{cl}
		& \E\Big[
		\normoneBig{
	\frac{1}{\sqrt{\nn}} \sum_{i\in\Ik}\big( \phi(\Si, \hetaIkc) - \E[ \phi(\Si, \hetaIkc)|\SIkc] \big) - 
	\frac{1}{\sqrt{\nn}} \sum_{i\in\Ik}\big( \phi(\Si, \etazero) - \E[ \phi(\Si, \etazero)] \big)
	}^2
		\Big|\SIkc\Big]\\
		=& 
		\frac{1}{\nn} \sum_{i\in\Ik}\E\big[ \big( \phi(\Si, \hetaIkc) - \phi(\Si, \etazero) \big)^2 \big|\SIkc\big]
		- \frac{1}{\nn} \sum_{i\in\Ik}\E[\phi(\Si, \hetaIkc) - \phi(\Si, \etazero) |\SIkc]^2\\
		&\quad + 
		\frac{1}{\nn}\sum_{i, j\in\Ik, i\neq j} \E\big[ \big(\phi(\Si,\hetaIkc)-\phi(\Si,\etazero)\big) \big(\phi(\Sj,\hetaIkc)-\phi(\Sj,\etazero)\big) \big| \SIkc\big]\\
		&\quad - 
		\frac{1}{\nn}\sum_{i, j\in\Ik, i\neq j} \E[ \phi(\Si,\hetaIkc)-\phi(\Si,\etazero)|\SIkc] \E[\phi(\Sj,\hetaIkc)-\phi(\Sj,\etazero)| \SIkc]. 
	\end{array}
\end{equation}
Let $i\in\indset{\NN}$. 
The nuisance parameter estimator $\hetaIkc$ belongs to $\TauN$ with $\PP$-probability   at least $1-\DeltaN$ by Assumption~\ref{assumpt:DML2}. 
Therefore, with $\PP$-probability at least $1-\DeltaN = 1-o(1)$, we have
\begin{displaymath}
	\begin{array}{cl}
		& \sqrt{\E\big[\big(\phi(\Si,\hetaIkc)-\phi(\Si,\etazero)\big)^2 \big|\SIkc\big]}\\
		\le & 
		\sup_{\eta\in\TauN} \normPBig{
		-\gonezero(\Di) + \gone(\Di) 
		+ \gzerozero(\Di) - \gzero(\Di) 
		+ \frac{\Wi}{\hzero(\Ui)}\epsYi \\
		&\quad
- \frac{\Wi}{\h(\Ui)}\big(\gonezero(\Di)-\gone(\Di) + \epsYi \big)
		- \frac{1-\Wi}{1-\hzero(\Ui)}\epsYi 
		+ \frac{1-\Wi}{1-\h(\Ui)}\big(\gzerozero(\Di) - \gzero(\Di) + \epsYi\big)}{2}\\
		\le & 
		\sup_{\eta\in\TauN}\normP{\gonezero(\Di) - \gone(\Di)}{2}
		+ \sup_{\eta\in\TauN}\normP{ \gzerozero(\Di) - \gzero(\Di) }{2}\\
		&\quad
		+ \sup_{\eta\in\TauN}\normPbig{\frac{\h(\Ui) - \hzero(\Ui)}{\hzero(\Ui)\h(\Ui)}\Wi\epsYi}{2}
		+ \sup_{\eta\in\TauN}\normPbig{\frac{\Wi}{\h(\Ui)} \big(\gonezero(\Di) - \gone(\Di)\big)}{2}\\
		&\quad
		+ \sup_{\eta\in\TauN}\normPbig{\frac{\hzero(\Ui) - \h(\Ui)}{(1-\hzero(\Ui))(1-\h(\Ui))}(1-\Wi)\epsYi}{2}
		+ \sup_{\eta\in\TauN}\normPbig{\frac{1-\Wi}{1-\h(\Ui)} \big(\gzerozero(\Di) - \gzero(\Di)\big)}{2}. 
	\end{array}
\end{displaymath}
Assumption~\ref{assumpt:regularity1}, \ref{assumpt:regularity3}, \ref{assumpt:regularity4}, and~\ref{assumpt:DML4} bound the terms $\normP{\Wi\epsYi/(\hzero(\Ui)\h(\Ui))}{\infty}$, $\normP{\Wi/\h(\Ui)}{\infty}$, $\normP{(1-\Wi)\epsYi/((1-\hzero(\Ui))(1-\h(\Ui)))}{\infty}$, and $\normP{(1-\Wi)/(1-\h(\Ui))}{\infty}$. 
 Assumption~\ref{assumpt:DML3} specifies that the error terms $\normP{\hzero(\Wi)-\h(\Wi)}{2}$, $\normP{\gonezero(\Di)-\gone(\Di)}{2}$, and $\normP{\gzerozero(\Di)-\gzero(\Di)}{2}$ are upper bounded by $\sqrt{\deltaN}\NN^{-\kappa}$. 
Due to H{\"o}lder's inequality,  
we infer
\begin{equation}\label{eq:bound-var-Ai}
	\sqrt{\E\big[\big(\phi(\Si,\hetaIkc)-\phi(\Si,\etazero)\big)^2 \big|\SIkc\big]}
	\lesssim \sqrt{\deltaN}\NN^{-\kappa}
\end{equation}
with $\PP$-probability at least $1-\DeltaN$. 

Subsequently, we bound the summands in~\eqref{eq:to-bound-vanishing-cov-lemma}. 
Due to~\eqref{eq:bound-var-Ai}, we have
\begin{displaymath}
	\frac{1}{\nn} \sum_{i\in\Ik}\E\big[ \big( \phi(\Si, \hetaIkc) - \phi(\Si, \etazero) \big)^2 \big|\SIkc\big]
		- \frac{1}{\nn} \sum_{i\in\Ik}\E[\phi(\Si, \hetaIkc) - \phi(\Si, \etazero) |\SIkc]^2
		\lesssim \deltaN\NN^{-2\kappa}
\end{displaymath}
with $\PP$-probability at least $1-\DeltaN$. 
Observe that we have
\begin{displaymath}
	\begin{array}{cl}
		&\frac{1}{\nn}\sum_{i, j\in\Ik, i\neq j} \E\big[ \big(\phi(\Si,\hetaIkc)-\phi(\Si,\etazero)\big) \big(\phi(\Sj,\hetaIkc)-\phi(\Sj,\etazero)\big) \big| \SIkc\big]\\
		&\quad - 
		\frac{1}{\nn}\sum_{i, j\in\Ik, i\neq j} \E[ \phi(\Si,\hetaIkc)-\phi(\Si,\etazero)|\SIkc] \E[\phi(\Sj,\hetaIkc)-\phi(\Sj,\etazero)| \SIkc]\\
		=& \frac{1}{\nn}\sum_{i, j\in\Ik, i\neq j} \Cov\big(\phi(\Si,\hetaIkc)-\phi(\Si,\etazero), \phi(\Sj,\hetaIkc)-\phi(\Sj,\etazero)\big|\SIkc\big)\\
		=& \frac{1}{\nn}\sum_{i, j\in\Ik, \{i, j\}\in \tildeE} \Cov\big(\phi(\Si,\hetaIkc)-\phi(\Si,\etazero), \phi(\Sj,\hetaIkc)-\phi(\Sj,\etazero)\big|\SIkc\big),
	\end{array}
\end{displaymath}
where $\tildeE$ denotes the edge set of the dependency graph, because the $\Si$ with $i\in\Ik$ are independent of data in $\SIkc$ and because, given $\SIkc$, $\phi(\Si\hetaIkc)-\phi(\Si,\etazero)$ and $\phi(\Sj,\hetaIkc)-\phi(\Sj,\etazero)$ are uncorrelated if there is no edge between $i$ and $j$ in the dependency graph. 
In the dependency graph, each node has a maximal degree of $\dmax$. Thus, there are at most $1/2\cdot \NN\cdot\dmax$ many edges in $\tildeE$. With $\PP$-probability at least $1-\DeltaN$, the term
\begin{displaymath}
	\Cov\big(\phi(\Si,\hetaIkc)-\phi(\Si,\etazero), \phi(\Sj,\hetaIkc)-\phi(\Sj,\etazero)\big|\SIkc\big)
\end{displaymath}
can be bounded by $\deltaN\NN^{-2\kappa}$ up to constants for all $i$ and $j$ due to~\eqref{eq:bound-var-Ai}. Therefore, with $\PP$-probability at least $1-\DeltaN$, we have
\begin{displaymath}
	\begin{array}{cl}
		&\frac{1}{\nn}\sum_{i, j\in\Ik, i\neq j} \E\big[ \big(\phi(\Si,\hetaIkc)-\phi(\Si,\etazero)\big) \big(\phi(\Sj,\hetaIkc)-\phi(\Sj,\etazero)\big) \big| \SIkc\big]\\
		&\quad - 
		\frac{1}{\nn}\sum_{i, j\in\Ik, i\neq j} \E[ \phi(\Si,\hetaIkc)-\phi(\Si,\etazero)|\SIkc] \E[\phi(\Sj,\hetaIkc)-\phi(\Sj,\etazero)| \SIkc]\\
		\lesssim& \deltaN \dmax\NN^{-2\kappa}\\
		\lesssim & \deltaN,
	\end{array}
\end{displaymath}
where the last bound holds due to Assumption~\ref{assumpt:DML3}. 
Consequently, we have
\begin{displaymath}
	\begin{array}{cl}
	& \E\Big[
		\normoneBig{
	\frac{1}{\sqrt{\nn}} \sum_{i\in\Ik}\big( \phi(\Si, \hetaIkc) - \E[ \phi(\Si, \hetaIkc)|\SIkc] \big) - 
	\frac{1}{\sqrt{\nn}} \sum_{i\in\Ik}\big( \phi(\Si, \etazero) - \E[ \phi(\Si, \etazero)] \big)
	}^2
		\Big|\SIkc\Big]\\
		\lesssim & \deltaN
	\end{array}
\end{displaymath}
with $\PP$-probability at least $1-\DeltaN$,
and we infer the statement of the lemma due to~\citet[Lemma 6.1]{Chernozhukov2018}.
\end{proof}

\begin{lemma}[Taylor expansion]\label{lem:Taylor}
Assume the assumptions of Theorem~\ref{thm:Gaussian} hold. 
	Let $k\in\indset{\KK}$. 
	We have
	\begin{displaymath}
		\normonebigg{\frac{1}{\sqrt{\nn}}\sum_{i\in\Ik}\big( \E[\phi(\Si,\hetaIkc)|\SIkc] - \E[\phi(\Si,\etazero)] \big)} = o_{\PP}(1)
	\end{displaymath}
\end{lemma}
\begin{proof}[Proof of Lemma~\ref{lem:Taylor}]
Let $k\in\indset{\KK}$. 
For $r\in [0, 1]$, let us define the function
\begin{displaymath}
	\fk(r) = \frac{1}{\nn}\sum_{i\in\Ik}\big( \E[\phi(\Si, \etazero + r(\hetaIkc-\etazero))|\SIkc] - \E[\phi(\Si,\etazero)]. 
\end{displaymath}
We have 
\begin{displaymath}
	\begin{array}{cl}
		&\E\Big[ \normoneBig{\frac{1}{\sqrt{\nn}}\sum_{i\in\Ik}\big( \E[\phi(\Si,\hetaIkc)|\SIkc] - \E[\phi(\Si,\etazero)] \big)}  \Big|\SIkc\Big]\\
		=& \normoneBig{\frac{1}{\sqrt{\nn}}\sum_{i\in\Ik}\big( \E[\phi(\Si,\hetaIkc)|\SIkc] - \E[\phi(\Si,\etazero)] \big)} \\
		=& \sqrt{\nn}\normone{\fk(1)}.
	\end{array}
\end{displaymath}
We apply a Taylor expansion to $\fk(1)$ at $0$ and obtain
\begin{displaymath}
	\fk(1) = \fk(0) + \fk'(0) + \frac{1}{2}\fk''(\tilde r)
\end{displaymath}
for some $\tilde r\in (0, 1)$. Thus, we have
\begin{displaymath}
	\E\Big[ \normoneBig{\frac{1}{\sqrt{\nn}}\sum_{i\in\Ik}\big( \E[\phi(\Si,\hetaIkc)|\SIkc] - \E[\phi(\Si,\etazero)] \big)}  \Big|\SIkc\Big]
	\le \sqrt{\nn}\Big( \normone{\fk(0)} + \normone{\fk'(0)} + \sup_{r\in(0, 1)}\frac{1}{2}\normone{\fk''(r)}\Big).
\end{displaymath}
Due to the definition of $\fk$, we have $\fk(0)=0$. Due to Neyman orthogonality that we established in Lemma~\ref{lem:Neyman-orth}, we have $\fk'(0)=0$. 
Due to the product property that we established in Lemma~\ref{lem:product-property}, we have $\sup_{r\in(0, 1)}\frac{1}{2}\normone{\fk''(r)}\lesssim\deltaN\NN^{-1/2}$ with $\PP$-probability at least $1-\DeltaN$ because $\hetaIkc$ belongs to $\TauN$ with $\PP$-probability at least $1-\DeltaN$. Consequently, we have 
\begin{displaymath}
	\E\Big[ \normoneBig{\frac{1}{\sqrt{\nn}}\sum_{i\in\Ik}\big( \E[\phi(\Si,\hetaIkc)|\SIkc] - \E[\phi(\Si,\etazero)] \big)}  \Big|\SIkc\Big]
	\lesssim \deltaN
\end{displaymath}
with $\PP$-probability at least $1-\DeltaN$. 
We infer the statement of the lemma due to~\citet[Lemma 6.1]{Chernozhukov2018}.
\end{proof}

\begin{proof}[Proof of Theorem~\ref{thm:Gaussian}]
We have
\begin{displaymath}
	\begin{array}{cl}
		&\sqrt{\NN}(\htheta-\thetazeroN) \\
		=& \sqrt{\NN}\cdot\frac{1}{\nn\KK}\sum_{\kk=1}^{\KK}\sum_{i\in\Ik}
		\psi(\Si, \thetazeroi, \hetaIkc) \\
		=& \frac{1}{\sqrt{\KK}}\sum_{\kk=1}^{\KK}\frac{1}{\sqrt{\nn}}\sum_{i\in\Ik}
		\big(\psi(\Si, \thetazeroi, \hetaIkc)-\psi(\Si, \thetazeroi, \etazero)\big)
		+ \frac{1}{\sqrt{\NN}}\sum_{i=1}^{\NN} \psi(\Si,\thetazeroi, \etazero)
	\end{array}
\end{displaymath}
because the disjoint sets $\Ik$ are of equal size $\nn$, so that we have $\NN=\nn\KK$. 
Let $\kk\in\indset{\KK}$. We have
\begin{displaymath}
	\begin{array}{cl}
		&\normoneBig{\frac{1}{\sqrt{\nn}}\sum_{i\in\Ik}\big( \psi(\Si,\thetazeroi, \hetaIkc) - \psi(\Si, \thetazeroi, \etazero) \big)}\\
		\le& \normoneBig{\frac{1}{\sqrt{\nn}}\sum_{i\in\Ik}\big( \psi(\Si,\thetazeroi, \hetaIkc) - \E[\psi(\Si,\thetazeroi, \hetaIkc)|\SIkc] \big)
		- 
		\frac{1}{\sqrt{\nn}}\sum_{i\in\Ik}\big( \psi(\Si,\thetazeroi, \etazero) - \E[\psi(\Si,\thetazeroi, \etazero)] \big)}\\
		&\quad + 
		\normoneBig{\frac{1}{\sqrt{\nn}}\sum_{i\in\Ik}\big(  \E[\psi(\Si,\thetazeroi, \hetaIkc)|\SIkc] - \E[\psi(\Si,\thetazeroi, \etazero)] \big)}\\
		=& o_{\PP}(1)
	\end{array}
\end{displaymath}
due to H{\"o}lder's inequality and Lemma~\ref{lem:vanishing-cov} and~\ref{lem:Taylor}. 
Because $\KK$ is a constant independent of $\NN$, we have
\begin{displaymath}
	 \frac{1}{\sqrt{\KK}}\sum_{\kk=1}^{\KK}\frac{1}{\sqrt{\nn}}\sum_{i\in\Ik}
		\big(\psi(\Si, \thetazeroi, \hetaIkc)-\psi(\Si, \thetazeroi, \etazero)\big)
		= o_{\PP}(1). 
\end{displaymath}
Due to Lemma~\ref{lem:Stein}, we have $\frac{1}{\sqrt{\NN}\cdot\sigmaN}\sum_{i=1}^{\NN} \psi(\Si,\thetazeroi, \etazero)\stackrel{d}{\rightarrow}\mathcal{N}(0, 1)$ as $\NN\to\infty$. 
Due to Assumption~\ref{assumpt:variance}, we therefore have 
\begin{displaymath}
		\frac{1}{\sqrt{\NN}\sigmainfty^{-1}}\sum_{i=1}^{\NN} \psi(\Si,\thetazeroi, \etazero)\\
		= \frac{1}{\sqrt{\NN}\cdot\sigmaN}\sum_{i=1}^{\NN} \psi(\Si,\thetazeroi, \etazero)\cdot \sigmaN\sigmainfty^{-1}
		\stackrel{d}{\rightarrow} \mathcal{N}(0, 1)
\end{displaymath}
as $\NN\to\infty$. Consequently, we have $\sqrt{\NN}\sigmainfty^{-1}(\htheta - \thetazeroN)\stackrel{d}{\rightarrow}\mathcal{N}(0, 1)$ as claimed. 
\end{proof}

\section{Bootstrap Variance Estimator}\label{sect:proofs-thm-var-est-boot}
We use the following assumption to establish the consistency of the bootstrap variance estimator. It is a high level assumption and
we will not verify it in terms of the model~\eqref{eq:SEM}; yet, assuming some form of continuity (as below) seems to be essentially necessary for the bootstrap to be consistent.

\begin{assumptions}\label{assumpt:var-boot}
To make the dependence of $\sigmainfty ^2$ in~\eqref{eq:product-property} on the law of the response error terms $\epsY$, the law of the covariates $C$, the nuisance functions $\etazero$, and the network $G$, we introduce the functional 
    \begin{displaymath}
        \T = 
        \lim_{\NN\rightarrow\infty}\Var\bigg(\frac{1}{\sqrt{\NN}}\sum_{i=1}^{\NN}\psi(\Si, \thetazeroi, \etazero)\bigg), 
    \end{displaymath}
    which can be represented as
    \begin{displaymath}
        \T = 
        \lim_{\NN\rightarrow\infty}\Var\bigg(\frac{1}{\sqrt{\NN}}\sum_{i=1}^{\NN}\phi(\Si, \etazero)\bigg)
    \end{displaymath}
    due to $\psi(\Si, \thetazeroi, \etazero) = \phi(\Si, \etazero)-\thetazeroi$ and because the $\thetazeroi$'s are non-random. 

 We assume that $\T$ is continuous with respect to Mallows' distance $d_2(\cdot,\cdot)$ in the first and second argument and with respect to $\normP{\cdot}{2}$ in the third argument. 
\end{assumptions}

\begin{proof}[Proof of Theorem~\ref{thm:var-est-boot}]
    The bootstrap variance relies on the same dependency structure induced by the network as $\sigmainfty^2$ and can be represented by 
     \begin{displaymath}
        \Tboot = 
        \lim_{\NN\rightarrow\infty}\Var^*\bigg(\frac{1}{\sqrt{\NN}}\sum_{i=1}^{\NN}\psi(\Si^*, \hthetazeroi, \heta)\bigg), 
    \end{displaymath}
    where the construction of $\Si^*$ is described in Section~\ref{sect:boot-var}. 
    Similarly to above, we can rewrite this bootstrap variance as 
    \begin{displaymath}
         \Tboot = 
        \lim_{\NN\rightarrow\infty}\Var^*\bigg(\frac{1}{\sqrt{\NN}}\sum_{i=1}^{\NN}\phi(\Si^*, \heta)\bigg). 
    \end{displaymath}
    Due to Assumption~\ref{assumpt:regularity1} and~\citep{BickelFreedman1981}, 
    we have $d_2(\hPC, \PC)\stackrel{P}{\rightarrow}0$, where $d_2(\cdot,\cdot)$ denotes Mallows' distance. Furthermore, due to $\normP{\hat{\varepsilon}_Y - \varepsilon_Y}{2}\stackrel{P}{\rightarrow} 0$, we also have $d_2(\hPheps,\Peps)\stackrel{P}{\rightarrow} 0$; see~\citep{BickelFreedman1981} and~\citep[Lemma 5.4]{Buehlmann1997}. 
    Due to $\normP{\hetaIkc-\etazero}{2}\stackrel{P}{\rightarrow} 0$ for $k\in\indset{\KK}$, we obtain 
    \begin{displaymath}
        \lim_{\NN\to\infty}
        \normonebigg{\Var^*\bigg(\frac{1}{\sqrt{\NN}}\sum_{i=1}^{\NN}\phi(\Si^*, \heta)\bigg) - 
        \Var\bigg(\frac{1}{\sqrt{\NN}}\sum_{i=1}^{\NN}\psi(\Si, \etazero)\bigg)}
        \stackrel{P}{\rightarrow} 0,
    \end{displaymath}
    which consequently establishes consistency of the bootstrap variance. 
\end{proof}

\section{Consistent Plugin Variance Estimator}\label{sect:var_est}

An alternative to the bootstrap variance estimator can be constructed as described below. We do not recommend this estimator unless the sample size is large relative to the network connectivity, but its consistency can be derived under different and more explicit conditions than in~\eqref{assumpt:var-boot}.

The challenge is that the unit-level effects $\thetazeroi$ for $i\in\indset{\NN}$ 
are not all equal. 
This is because the unit-level data points $\Si$ are typically not identically distributed. The difference in distributions originates from the $X$- and $Z$- features that generally depend on a varying number of other units. If two unit-level data points $\Si$ and $\Sj$ have the same distribution, then their unit-level treatment effects $\thetazeroi$ and $\thetazeroj$ coincide. 
If enough of these unit-level treatment effects coincide, we can use the corresponding unit-level data to estimate them.
Subsequently, we describe this procedure. 

We partition $\indset{\NN}$ into sets $\Acald$ for $d\ge 0$ such that all unit-level data points $\Si$ for $i\in\Acald$ have the same distribution.
Provided that the sets $\Acald$ are large enough, we can consistently estimate the corresponding $\thetazerod$ for $d\ge 0$ by 
\begin{equation}\label{eq:theta-d-est}
	\hthetad = \frac{1}{\normone{\Acald}}\sum_{i\in\Acald} \phi(\Si, \hetaIkcdegr{i}), 
\end{equation}
where $\kk(i)$ denotes the index in $\indset{\KK}$ such that $i\in\Ikdegr{i}$. The convergence rate of these estimators is at least $\NN^{-1/4}$; see Lemma~\ref{lem:hthetad} in Section~\ref{sect:proofs-thm-var-est} in the appendix. 
To achieve this rate, we require that the sets $\Acald$ contain at least of order $\NN^{3/4}$ many indices; see Assumption~\ref{assumpt:Acald} in Section~\ref{sect:AssumptionsDefinitions} in the appendix. 
The parametric convergence rate cannot be achieved in general because $\Acald$ is of smaller size than $\NN$, but the corresponding units may have the maximal $\dmax$ many ties in the network. 

Subsequently, we characterize a situation in which the index $d$ corresponds to the degree in the dependency graph $\tildeG$. This is the case if two unit-level data points $S_i$ and $S_j$ have the same distribution if and only if the units $i$ and $j$ have the same degree in $\tildeG$. We assume, given a unit $i$ and some $m \in \indset{N}\setminus \{i\}$, that $1)$ if $\Cm$ is part of $\Zi$, then $\Cm$ is also part of $\Xi$ and vice versa; and $2)$ if $\Wm$ is part of $\Xi$, then $\Cm$ is part of $\Xi$ and $\Zi$ and vice versa. 
Consequently, if two units $i\neq j$ have the same degree in the dependency graph, then their $X$- and their $Z$-features are computed using the same number of random variables. Hence, $X_i$ and $X_j$ as well as $Z_i$ and $Z_j$ are identically distributed, and therefore $S_i$ and $S_j$ have the same distribution.
Thus, the sets $\Acald$ form partition of the units according to their degree in the dependency graph, that is, 
	$\Acald = \{i\in\indset{\NN} \colon \degr{i} = d\}$ 
for $d\ge 0$, where $\degr{i}$ denotes the degree of $i$ in the dependency graph. There are $\dmax + 1 = o(\NN^{1/4})$ 
many such sets, and each of them is required to be of size at least of order $\NN^{3/4}$ in Lemma~\ref{lem:hthetad}. This is feasible because there are $\NN$ units in total. 
Provided that the machine learning estimators of the nuisance functions converge at a rate faster than $\NN^{1/4}$ as specified by 
Assumption~\ref{assumpt:kappa} in the appendix, we have the following consistent estimator of the asymptotic variance given in Theorem~\ref{thm:var-est}. 
Algorithm \ref{algo:Summary} summarizes the whole  procedure of point estimation and inference for the EATE where the variance is estimated as given in Theorem~\ref{thm:var-est}. Nevertheless, this estimation scheme can be extended to general sets $\Acald$. 

\begin{theorem}\label{thm:var-est}
	Denote by $\tildeG=(V, \tildeE)$ the dependency graph on $\Si$, $i\in\indset{\NN}$. For a unit $i\in\indset{\NN}$, denote by $\degr{i}$ its degree in $\tildeG$ and by $\kk(i)$ the number in $\indset{\KK}$ such that $\Si\in\Ikdegr{i}$. In addition to the assumptions made in Theorem~\ref{thm:Gaussian}, also assume that Assumption~\ref{assumpt:Acald} and~\ref{assumpt:kappa} stated in Section~\ref{sect:AssumptionsDefinitions} in the appendix hold. 
		Based on $\phi$ defined in \eqref{eq:score}, we define the score function $\psi(\Si, \theta, \eta) = \phi(\Si,\eta) - \theta$
for some general $\theta\in\R$ and the nuisance function triple $\eta = (\gone, \gzero, \h)$. 
	Then, 
	\begin{displaymath}
		\frac{1}{\NN}\sum_{i=1}^{\NN}\psi^2(\Si, \hthetadegr{i}, \hetaIkcdegr{i})
		+ \frac{2}{\NN}\sum_{\{i, j\}\in \tildeE} \psi(\Si, \hthetadegr{i}, \hetaIkcdegr{i})\psi(\Sj, \hthetadegr{j}, \hetaIkcdegr{j})
	\end{displaymath}
is a consistent estimator of the asymptotic variance $\sigmainfty^2$ in Theorem~\ref{thm:Gaussian}. 
\end{theorem}

\subsection{Proof of Theorem~\ref{thm:var-est}}\label{sect:proofs-thm-var-est}

\begin{lemma}\label{lem:boundLpfournorm}
Assume the assumptions of Theorem~\ref{thm:var-est} hold.
Let $i\in\indset{\NN}$. 
There exists a finite real constant $\CLpfournorm$ independent of $i$ such that $\normP{\psi(\Si, \thetazerodegr{i},\etazero)}{4} \le \CLpfournorm$ holds. 
Consequently, for $i, j, m, r\in\indset{\NN}$, we can also bound the following terms by finite uniform constants:
\begin{itemize}
	\item $\normP{\psi(\Si, \thetazerodegr{i}, \etazero)}{2}$
	\item $\Var\big(\phi(\Si,\etazero)\big)$
	\item $\Var\big(\psi^2(\Si, \thetazerodegr{i}, \etazero)\big)$
	\item $\Cov\big(\phi(\Si,\etazero), \phi(\Sj, \etazero)\big)$
	\item $\Var\big( \psi(\Si, \thetazerodegr{i},\etazero)\psi(\Sj, \thetazerodegr{j},\etazero) \big)$
	\item $\Cov\big(\psi(\Si, \thetazerodegr{i},\etazero)\psi(\Sj, \thetazerodegr{j},\etazero) , \psi(\Sm, \thetazerodegr{m},\etazero)\psi(\Sr, \thetazerodegr{r},\etazero) \big)$
\end{itemize}
Moreover, we have $\phi^2(\Si,\etazero) = O_{\PP}(1)$. Furthermore, we have $\psi^2(\Si, \thetazerodegr{i}, \hetaIkcdegr{i}) = O_{\PP}(1)$. 
\end{lemma}

\begin{proof}[Proof of Lemma~\ref{lem:boundLpfournorm}]
We have
\begin{equation}\label{eq:helper1}
	\begin{array}{cl}
		& \normP{\psi(\Si, \thetazerodegr{i},\etazero)}{4} \\
		\le & 
		\normP{\gonezero(\Di)}{4}
		+ \normP{\gzerozero(\Di)}{4}
		+ \normPBig{\frac{\Wi}{\hzero(\Ui)}}{4}\normP{\Yi - \gonezero(\Di)}{\infty}\\
		&\quad
		+ \normPBig{\frac{1-\Wi}{1-\hzero(\Ui)}}{4}\normP{\Yi - \gzerozero(\Di)}{\infty}
		+ \normone{\thetazerodegr{i}}. 
	\end{array}
\end{equation}
All individual summands in the above decomposition are bounded by a finite real constant independent of $i$ due to Assumption~\ref{assumpt:regularity}. Therefore, there exists a finite real constant $\CLpfournorm$ independent of $i$ such that $\normP{\psi(\Si, \thetazeroi,\etazero)}{4} \le \CLpfournorm$ holds. 

The other terms in the statement of the present lemma are bounded as well by finite real constants independent of $i, j, m, r\in\indset{\NN}$ due to H{\"o}lder's inequality. 

Moreover, we have $\psi^2(\Si,\etazero) = O_{\PP}(1)$ because $\normP{\psi^2(\Si,\etazero)}{2}$ is bounded by a constant that is independent of $i$. 

Furthermore, with $\PP$-probability at least $1-\DeltaN$, we have
\begin{displaymath}
		\E\big[ \psi^2(\Si, \thetazerodegr{i}, \hetaIkcdegr{i}) \big|\SIkcdegr{i}\big] 
		\le  \sup_{\eta\in\TauN} 
		\E\big[ \psi^2(\Si, \thetazerodegr{i}, \eta)\big] 
		=   \sup_{\eta\in\TauN} \normP{ \psi(\Si, \thetazerodegr{i}, \eta)}{2}^2. 
\end{displaymath}
The term $\normP{ \psi(\Si, \thetazerodegr{i}, \eta)}{2}^2$ is bounded by a real constant that is independent of $i$ and $\eta$ because the derivation in~\eqref{eq:helper1} also holds with $\etazero$ replaced by $\eta\in\TauN$ due to Assumption~\ref{assumpt:DML}. 
\end{proof}

\begin{lemma}[Convergence rate of unit-level  effect estimators]\label{lem:hthetad}
Assume the assumptions of Theorem~\ref{thm:var-est} hold.
	Let $d\ge 0$, 
	and assume that all assumptions of Section~\ref{sect:AssumptionsDefinitions} in the appendix hold. 
	Then, we have $\hthetad - \thetazerod = o_{\PP}(\NN^{-1/4})$, where $\hthetad$ is as in~\eqref{eq:theta-d-est}.
\end{lemma}
\begin{proof}[Proof of Lemma~\ref{lem:hthetad}]
Let $d\ge 0$. 
Due to the definition of $\hthetad$ given in~\eqref{eq:theta-d-est} and Lemma~\ref{lem:identifiability}, we have
\begin{equation}\label{eq:toBound1}
	\begin{array}{cl}
		& \NN^{\frac{1}{4}}(\hthetad - \thetazerod)\\
		=& \frac{\NN^{\frac{1}{4}}}{\normone{\Acald}}\sum_{i\in\Acald} \big(\phi(\Si, \hetaIkcdegr{i}) -\E[\phi(\Si, \etazero)] \big)\\
		=& \frac{\NN^{\frac{1}{4}}}{\normone{\Acald}}\sum_{i\in\Acald} \big(\phi(\Si, \hetaIkcdegr{i}) -\phi(\Si, \etazero) \big)
		+ \frac{\NN^{\frac{1}{4}}}{\normone{\Acald}}\sum_{i\in\Acald} \big(\phi(\Si, \etazero) -\E[\phi(\Si, \etazero)] \big). 
	\end{array}
\end{equation}
Subsequently, we show that the two sets of summands in~\eqref{eq:toBound1} are of order $o_{\PP}(1)$. We start with the first set of summands. Let $i\in\Acald$. With $\PP$-probability at least $1-\DeltaN$, we have
\begin{displaymath}
	\sqrt{\E\big[\big(\phi(\Si,\hetaIkcdegr{i})-\phi(\Si,\etazero)\big)^2 \big|\SIkc\big]}
	\lesssim \sqrt{\deltaN}\NN^{-\kappa}
\end{displaymath}
due to Equation~\eqref{eq:bound-var-Ai}. Hence, we have $\normone{\phi(\Si,\hetaIkcdegr{i})-\phi(\Si,\etazero)} = O_{\PP}(\sqrt{\deltaN}\NN^{-\kappa})$
due to~\citet[Lemma 6.1]{Chernozhukov2018}. 
Consequently, we have 
\begin{displaymath}
	\frac{\NN^{\frac{1}{4}}}{\normone{\Acald}}\sum_{i\in\Acald} \normone{\phi(\Si, \hetaIkcdegr{i}) -\phi(\Si, \etazero)} 
	=  O_{\PP}(\sqrt{\deltaN}\NN^{\frac{1}{4}-\kappa}) 
	= o_{\PP}(1)
\end{displaymath}
because we have $\kappa\ge 1/4$ by Assumption~\ref{assumpt:kappa}. 
Next, we show that the second set of summands in~\eqref{eq:toBound1} is of order $o_{\PP}(1)$. 
Let $\eps>0$. We have
\begin{displaymath}
	\begin{array}{cl}
		& \PP \Big( \normoneBig{\frac{\NN^{\frac{1}{4}}}{\normone{\Acald}}\sum_{i\in\Acald} \big(\phi(\Si, \etazero) -\E[\phi(\Si, \etazero)] \big)}^2 > \eps^2 \Big)\\
		\le & \frac{\NN^{\frac{1}{2}}}{\eps^2\normone{\Acald}^2}
		\Big(\sum_{i\in\Acald}\Var(\phi(\Si, \etazero))
		+ \sum_{i, j\in\Acald, i\neq j} \Cov\big(\phi(\Si,\etazero), \phi(\Si, \etazero)\big)
		\Big)\\
		=& \frac{\NN^{\frac{1}{2}}}{\eps^2\normone{\Acald}^2}(\normone{\Acald} + 2\normone{\tildeE\cap\Acald^2})O(1)
	\end{array}
\end{displaymath}
because $\Var(\phi(\Si, \etazero))$ and $\Cov(\phi(\Si,\etazero), \phi(\Si, \etazero))
$ are bounded by constants uniformly over $i$ due to Lemma~\ref{lem:boundLpfournorm}, and because $\Cov(\phi(\Si,\etazero), \phi(\Si, \etazero))
$ does not equal $0$ only if $\{i,j\}\in \tildeE\cap\Acald^2$, where $\tildeE$ denotes the edge set of the dependency graph. There are $\normone{\Acald}$ many nodes in $\Acald$, and each node has a maximal degree of $\dmax$. 
Thus, we have $\normone{\tildeE\cap\Acald^2}\le 1/2\normone{\Acald}\dmax$. Due to $\dmax = o(\NN^{1/4})$ and $\normone{\Acald} = \Omega(\NN^{3/4})$, which hold according to Assumption~\ref{assumpt:degree} and~\ref{assumpt:Acald}, we obtain
\begin{displaymath}
	\frac{\NN^{\frac{1}{2}}}{\eps^2\normone{\Acald}^2}(\normone{\Acald} + 2\normone{E\cap\Acald^2})O(1)
	= o(1).
\end{displaymath}
Consequently, we also have 
\begin{displaymath}
\normonebigg{\frac{\NN^{\frac{1}{4}}}{\normone{\Acald}}\sum_{i\in\Acald} \big(\phi(\Si, \etazero) -\E[\phi(\Si, \etazero)] \big)} 
= o_{\PP}(1). 
\end{displaymath}
\end{proof}

\begin{lemma}[Consistent variance estimator part I]\label{lem:consistent-variance-part-I}
Assume the assumptions of Theorem~\ref{thm:var-est} hold.
	We have
	\begin{displaymath}
		\normonebigg{
		\frac{1}{\NN}\sum_{i=1}^{\NN} \big(\psi^2(\Si, \hthetadegr{i}, \hetaIkcdegr{i}) -\E[\psi^2(\Si,\thetazerodegr{i},\etazero)] \big)
		}
		= o_{\PP}(1).
	\end{displaymath}
\end{lemma}
\begin{proof}[Proof of Lemma~\ref{lem:consistent-variance-part-I}]
We have
\begin{equation}\label{eq:toBound2}
	\begin{array}{cl}
		& \frac{1}{\NN}\sum_{i=1}^{\NN} \big(\psi^2(\Si, \hthetadegr{i}, \hetaIkcdegr{i}) -\E[\psi^2(\Si,\thetazerodegr{i},\etazero)] \big)\\
		=& 
		\frac{1}{\NN}\sum_{i=1}^{\NN} \big(\psi^2(\Si, \hthetadegr{i}, \hetaIkcdegr{i}) -\psi^2(\Si,\hthetadegr{i},\etazero) \big)\\
		&\quad
		+ \frac{1}{\NN}\sum_{i=1}^{\NN} \big(\psi^2(\Si,\hthetadegr{i},\etazero) - \psi^2(\Si,\thetazerodegr{i},\etazero) \big)\\
		&\quad
		+  \frac{1}{\NN}\sum_{i=1}^{\NN} \big(\psi^2(\Si,\thetazerodegr{i},\etazero) - \E[\psi^2(\Si,\thetazerodegr{i},\etazero)]\big). 
	\end{array}
\end{equation}
We bound the three sets of summands in~\eqref{eq:toBound2} individually. 
The first set of summands can be expressed as
\begin{displaymath}
	\begin{array}{cl}
		& \frac{1}{\NN}\sum_{i=1}^{\NN} \big(\psi^2(\Si, \hthetadegr{i}, \hetaIkcdegr{i}) -\psi^2(\Si,\hthetadegr{i},\etazero) \big)\\
		=&
		\frac{1}{\NN}\sum_{i=1}^{\NN} \big( \phi^2(\Si, \hetaIkcdegr{i}) - \phi^2(\Si,\etazero) \big)
		- \frac{2}{\NN} \sum_{i=1}^{\NN} \hthetadegr{i}\big( \phi(\Si,\hetaIkcdegr{i}) - \phi(\Si,\etazero) \big). 
	\end{array}
\end{displaymath}
We have 
\begin{equation}\label{eq:toBound3}
	\normonebigg{\frac{1}{\NN}\sum_{i=1}^{\NN} \big( \phi^2(\Si, \hetaIkcdegr{i}) - \phi^2(\Si,\etazero) \big)} = o_{\PP}(1)
\end{equation}
because the function $\R\ni x\mapsto x^2\in \R$ is continuous and due to Equation~\eqref{eq:bound-var-Ai}. Indeed, let $\eps>0$. Because the function $\R\ni x\mapsto x^2\in \R$ is continuous, there exists $\delta>0$ such that if $\normone{\phi(\Si, \hetaIkcdegr{i}) - \phi(\Si,\etazero)}<\delta$, then also $\normone{\phi^2(\Si, \hetaIkcdegr{i}) - \phi^2(\Si,\etazero)}<\eps$. Consequently, 
we have
\begin{displaymath}
	\begin{array}{cl}
		& \PP\big( \normone{\phi^2(\Si, \hetaIkcdegr{i}) - \phi^2(\Si,\etazero)}>\eps \big|\SIkcdegr{i}\big)\\
		\le &  \PP\big( \normone{\phi(\Si, \hetaIkcdegr{i}) - \phi(\Si,\etazero)}>\delta \big|\SIkcdegr{i}\big)\\
		\le & \frac{1}{\delta}\sup_{\eta\in\TauN} \normP{\phi(\Si, \eta) - \phi(\Si,\etazero)}{1}
	\end{array}
\end{displaymath}
with $\PP$-probability at least $1-\DeltaN$, and we infer~\eqref{eq:toBound3} due to~\eqref{eq:bound-var-Ai}. 
The estimator $\hthetadegr{i}$ is a consistent estimator of $\thetazerodegr{i}$ due to Lemma~\ref{lem:hthetad}, and $\thetazerodegr{i}$ is bounded independent of $i$ due to Assumption~\ref{assumpt:regularity5}. Moreover, we have $\normone{\phi(\Si,\hetaIkcdegr{i}) - \phi(\Si,\etazero)}=o_{\PP}(1)$ due to~\eqref{eq:bound-var-Ai} and~\citet[Lemma 6.1]{Chernozhukov2018}. Consequently, we have
\begin{displaymath}
	\normonebigg{\frac{2}{\NN} \sum_{i=1}^{\NN} \hthetadegr{i}\big( \phi(\Si,\hetaIkcdegr{i}) - \phi(\Si,\etazero) \big)} = o_{\PP}(1)
\end{displaymath}
due to H{\"o}lder's inequality. Hence, the first set of summands in~\eqref{eq:toBound2} is of order $o_{\PP}(1)$. 
The second set of summand in~\eqref{eq:toBound2} can be decomposed as
\begin{displaymath}	
	\begin{array}{cl}
		&\frac{1}{\NN}\sum_{i=1}^{\NN} \big(\psi^2(\Si,\hthetadegr{i},\etazero) - \psi^2(\Si,\thetazerodegr{i},\etazero) \big)\\
		=& \frac{1}{\NN}\sum_{i=1}^{\NN}(\hthetadegr{i}^2 - (\thetazerodegr{i})^2)
		- \frac{2}{\NN}\sum_{i=1}^{\NN} (\hthetadegr{i} - \thetazerodegr{i})\phi(\Si,\etazero). 
	\end{array}
\end{displaymath}
We have $\normone{\frac{1}{\NN}\sum_{i=1}^{\NN}(\hthetadegr{i}^2 - (\thetazerodegr{i})^2)}=o_{\PP}(1)$ due to Lemma~\ref{lem:hthetad}. Lemma~\ref{lem:boundLpfournorm} bounds $\phi^2(\Si,\etazero)$ in probability.  Due to H{\"o}lder's inequality, we obtain
\begin{displaymath}
	\normonebigg{\frac{2}{\NN}\sum_{i=1}^{\NN} (\hthetadegr{i} - \thetazerodegr{i})\phi(\Si,\etazero)}
	= o_{\PP}(1).
\end{displaymath}
Consequently, the second set of summands in~\eqref{eq:toBound2} is of order $o_{\PP}(1)$. Last, we bound the third set of summands in~\eqref{eq:toBound2}. Let $\eps > 0$. We have
\begin{displaymath}
	\begin{array}{cl}
		& \PP\Big( \normoneBig{\frac{1}{\NN}\sum_{i=1}^{\NN} \big(\psi^2(\Si,\thetazerodegr{i},\etazero) - \E[\psi^2(\Si,\thetazerodegr{i},\etazero)]\big)}^2>\eps^2 \Big)\\
		\le & 
		\frac{1}{\eps^2\NN^2}\Big(
		\sum_{i=1}^{\NN}\Var\big(\psi^2(\Si,\thetazerodegr{i},\etazero)\big)
		+ \sum_{i,j\in\indset{\NN}, \{i,j\}\in \tildeE} \Cov\big(\psi^2(\Si,\thetazerodegr{i},\etazero), \psi^2(\Sj,\thetazerodegr{j},\etazero)\big)
		\Big)\\
		\le & \frac{1}{\eps^2\NN^2}(\NN O(1) + \NN\dmax O(1))\\
		=& o(1)
	\end{array}
\end{displaymath}
because $\Var(\psi^2(\Si,\thetazerodegr{i},\etazero))$ and $\Cov(\psi^2(\Si,\thetazerodegr{i},\etazero), \psi^2(\Sj,\thetazerodegr{j},\etazero))$ are bounded uniformly over $i$ and $j$ by Lemma~\ref{lem:boundLpfournorm}, because $\Cov(\psi^2(\Si,\thetazerodegr{i},\etazero), \psi^2(\Sj,\thetazerodegr{j},\etazero))$ does not vanish only if $\{i,j\}\in \tildeE$, and because $\dmax=o(\NN^{1/4})$ by Assumption~\ref{assumpt:degree}. 
Consequently, also the third set of summands in~\eqref{eq:toBound2} is of order $o_{\PP}(1)$, and we have established the statement of the present lemma. 
\end{proof}

\begin{lemma}[Consistent variance estimator part II]\label{lem:consistent-variance-part-II}
Assume the assumptions of Theorem~\ref{thm:var-est} hold.
	Denote by $\tildeE$ the edge set of the dependency graph. 
	We have
	\begin{displaymath}
		\normonebigg{
		\frac{1}{\NN}\sum_{i, j \in\indset{\NN}, \{i,j\}\in \tildeE} \big( 
		\psi(\Si, \hthetadegr{i}, \hetaIkcdegr{i})\psi(\Sj, \hthetadegr{j}, \hetaIkcdegr{j})
		- \E[\psi(\Si,\thetazerodegr{i},\etazero)\psi(\Sj,\thetazerodegr{j},\etazero)]
		\big)
		}
		= o_{\PP}(1). 
	\end{displaymath}
\end{lemma}
\begin{proof}[Proof of Lemma~\ref{lem:consistent-variance-part-II}]
We have the decomposition
\begin{equation}\label{eq:toBound4}
	\begin{array}{cl}
		&\frac{1}{\NN}\sum_{i, j \in\indset{\NN}, \{i,j\}\in \tildeE} \big( 
		\psi(\Si, \hthetadegr{i}, \hetaIkcdegr{i})\psi(\Sj, \hthetadegr{j}, \hetaIkcdegr{j})
		- \E[\psi(\Si,\thetazerodegr{i},\etazero)\psi(\Sj,\thetazerodegr{j},\etazero)]
		\big)\\
		=& \frac{2}{\NN}\sum_{\{i,j\}\in \tildeE} \big( 
		\psi(\Si, \hthetadegr{i}, \hetaIkcdegr{i})\psi(\Sj, \hthetadegr{j}, \hetaIkcdegr{j})
		- \psi(\Si, \thetazerodegr{i}, \hetaIkcdegr{i})\psi(\Sj, \thetazerodegr{j}, \hetaIkcdegr{j})
		\big)\\
		&\quad + 
		\frac{2}{\NN}\sum_{\{i,j\}\in \tildeE} \big( 
		\psi(\Si, \thetazerodegr{i}, \hetaIkcdegr{i})\psi(\Sj, \thetazerodegr{j}, \hetaIkcdegr{j})
		- \psi(\Si, \thetazerodegr{i}, \etazero)\psi(\Sj, \thetazerodegr{j}, \etazero)
		\big)\\
		&\quad + 
		\frac{2}{\NN}\sum_{\{i,j\}\in \tildeE} \big( 
		\psi(\Si, \thetazerodegr{i}, \etazero)\psi(\Sj, \thetazerodegr{j}, \etazero)
		- \E[\psi(\Si, \thetazerodegr{i}, \etazero)\psi(\Sj, \thetazerodegr{j}, \etazero)]
		\big). 
	\end{array}
\end{equation}
Subsequently, we bound the three sets of summands in~\eqref{eq:toBound4} individually. We start by bounding the first set of summands. 
We have
\begin{displaymath}	
	\begin{array}{cl}
		& \frac{1}{\NN}\sum_{\{i,j\}\in \tildeE} \big( 
		\psi(\Si, \hthetadegr{i}, \hetaIkcdegr{i})\psi(\Sj, \hthetadegr{j}, \hetaIkcdegr{j})
		- \psi(\Si, \thetazerodegr{i}, \hetaIkcdegr{i})\psi(\Sj, \thetazerodegr{j}, \hetaIkcdegr{j})
		\big)\\
		=& \frac{2}{\NN}\sum_{\{i, j\}\in \tildeE}(\thetazerodegr{i} - \hthetadegr{i}) \psi(\Sj, \thetazerodegr{j}, \hetaIkcdegr{j})
		+ \frac{1}{\NN}\sum_{\{i, j\}\in \tildeE}(\thetazerodegr{i} - \hthetadegr{i})(\thetazerodegr{j} - \hthetadegr{j}). 
	\end{array}
\end{displaymath}
We have
\begin{displaymath}
	\begin{array}{cl}
	& \normonebigg{\frac{1}{\NN}\sum_{\{i, j\}\in \tildeE}(\thetazerodegr{i} - \hthetadegr{i}) \psi(\Sj, \thetazerodegr{j}, \hetaIkcdegr{j})}\\
	\le& \sqrt{\frac{1}{\NN}\sum_{\{i, j\}\in \tildeE}(\thetazerodegr{i} - \hthetadegr{i})^2}
	\sqrt{\frac{1}{\NN}\sum_{\{i, j\}\in \tildeE} \psi(\Sj, \thetazerodegr{j}, \hetaIkcdegr{j})}\\
	=& \frac{1}{\NN}\normone{\tildeE}o_{\PP}(\NN^{-1/4}) \\
	=& \dmax o_{\PP}(\NN^{-1/4})\\
	=& o_{\PP}(1)
	\end{array}
\end{displaymath}
due to H{\"o}lder's inequality, Lemma~\ref{lem:hthetad}, Lemma~\ref{lem:boundLpfournorm}, and Assumption~\ref{assumpt:degree}. 
Moreover, we have 
\begin{displaymath}
	\normonebigg{\frac{1}{\NN}\sum_{\{i, j\}\in \tildeE}(\thetazerodegr{i} - \hthetadegr{i})(\thetazerodegr{j} - \hthetadegr{j})}
	= \frac{1}{\NN}\normone{\tildeE} o_{\PP}(\NN^{-1/2}) = o_{\PP}(1)
\end{displaymath}
due to H{\"o}lder's inequality, Lemma~\ref{lem:hthetad}, and Assumption~\ref{assumpt:degree}. 
Consequently, the first set of summands in~\eqref{eq:toBound4} is of order $o_{\PP}(1)$. We proceed to bound the second set of summands in~\eqref{eq:toBound4}. 
Let $\{i,j\}\in \tildeE$. Due to the construction of $\SIkdegr{i}$ and $\SIkcdegr{i}$, we have $\Si = (\Wi, \CXZi, \Yi)\in\SIkdegr{i}$, and none of $\Wi$, $\Ci$, $\Yi$, or the variables used to compute $\Xi$ belong to $\SIkcdegr{i}$. Moreover, the variables $\Wi$, $\Ci$, $\Yi$, and the variables used to compute  $\Xi$ also cannot belong to $\SIkcdegr{j}$ as otherwise we would have $\Si\independent\Sj$, and consequently $\{i,j\}\not\in \tildeE$. 
Therefore, we have
\begin{displaymath}
	\begin{array}{cl}
		& \E\big[ \normone{\psi(\Si, \thetazerodegr{i}, \hetaIkcdegr{i})\psi(\Sj, \thetazerodegr{j}, \hetaIkcdegr{j})
		- \psi(\Si, \thetazerodegr{i}, \etazero)\psi(\Sj, \thetazerodegr{j}, \etazero)}
		\big|\SIkcdegr{i}, \SIkcdegr{j} \big]\\
		\le & \sup_{\eta_1,\eta_2\in\TauN} 
		\E\big[ \normone{\psi(\Si, \thetazerodegr{i}, \eta_1)\psi(\Sj, \thetazerodegr{j}, \eta_2)
		- \psi(\Si, \thetazerodegr{i}, \etazero)\psi(\Sj, \thetazerodegr{j}, \etazero)}
		\big]\\
		\le & 
		\sup_{\eta_1\in\TauN}\normP{\phi(\Si,\eta_1) - \phi(\Si,\etazero)}{2}\normP{\psi(\Sj, \thetazerodegr{j}, \etazero)}{2}\\
		&\quad 
		+ \sup_{\eta_2\in\TauN} \normP{\psi(\Si, \thetazerodegr{i}, \etazero)}{2} \normP{\phi(\Sj,\eta_2) - \phi(\Sj,\etazero)}{2}\\
		&\quad 
		+ \sup_{\eta_1,\eta_2\in\TauN}\normP{\phi(\Si,\eta_1) - \phi(\Si,\etazero)}{2} \normP{\phi(\Sj,\eta_2) - \phi(\Sj,\etazero)}{2}

	\end{array}
\end{displaymath}
with $\PP$-probability at least $1-\DeltaN$ due to H{\"o}lder's inequality. Because all terms above are uniformly bounded due to Lemma~\ref{lem:boundLpfournorm}, we infer that the second set of summands in~\eqref{eq:toBound4} is of order $o_{\PP}(1)$ due to~\citet[Lemma 6.1]{Chernozhukov2018}. Finally, we bound the third set of summands in~\eqref{eq:toBound4}. Let $\eps>0$. We have
\begin{equation}\label{eq:toBound5}
	\begin{array}{cl}
		& \PP\Big( \normoneBig{\frac{1}{\NN}\sum_{\{i,j\}\in \tildeE} \big( 
		\psi(\Si, \thetazerodegr{i}, \etazero)\psi(\Sj, \thetazerodegr{j}, \etazero)
		- \E[\psi(\Si, \thetazerodegr{i}, \etazero)\psi(\Sj, \thetazerodegr{j}, \etazero)]
		\big)}^2 > \eps^2 \Big)\\
		\le& 
		\frac{1}{\eps^2\NN^2} \Big(
		\sum_{\{i,j\}\in \tildeE}\Var\big( \psi(\Si,\thetazerodegr{i}, \etazero)\psi(\Si,\thetazerodegr{j}, \etazero) \big)\\
		&\quad + 
		\sum_{\{i,j\}, \{m, r\} \in \tildeE, \mathrm{unequal}}
		\Cov\big(
		\psi(\Si,\thetazerodegr{i}, \etazero)\psi(\Sj,\thetazerodegr{j}, \etazero), 
		\psi(\Sm,\thetazerodegr{m}, \etazero)\psi(\Sr,\thetazerodegr{r}, \etazero)
		\big)
		\Big). 
	\end{array}
\end{equation}
Due to Lemma~\ref{lem:boundLpfournorm}, 
the variance and covariance terms in~\eqref{eq:toBound5}
are uniformly bounded by constants. Furthermore, the covariance terms do only not equal $0$ if $\Si$ depends on $\Sm$ or $\Sr$, or if $\Sj$ depends on $\Sm$ or $\Sr$. In order to better describe these dependency relationships, we build a graph on the edge set of the dependency graph. We consider the graph $ G'=( V',  E')$ with $ V' = \tildeE$ and such that an edge $\{\{i,j\},\{m,r\}\}\in E'$ if and only if at least one of $\{i,m\}$, $\{i,r\}$, $\{j,m\}$, $\{j,r\}$ 
belongs to $\tildeE$. Consequently, $\{\{i,j\},\{m,r\}\}\in E'$ if and only if $(\Si, \Sj) \not\independent (\Sm, \Sr)$, in which case the covariance term in~\eqref{eq:toBound5} corresponding to $\{i,j\}$ and $\{m,r\}$ does not vanish. Furthermore, we have $\normone{ E'} = 1/2\normone{\tildeE} \tildedmax$, where $\tildedmax$ denotes the maximal degree of a node in $ G'$. We have $\tildedmax\le 2\dmax$. Consequently, we have 
\begin{displaymath}
	\begin{array}{cl}
		& \PP\Big( \normoneBig{\frac{1}{\NN}\sum_{\{i,j\}\in \tildeE} \big( 
		\psi(\Si, \thetazerodegr{i}, \etazero)\psi(\Sj, \thetazerodegr{j}, \etazero)
		- \E[\psi(\Si, \thetazerodegr{i}, \etazero)\psi(\Sj, \thetazerodegr{j}, \etazero)]
		\big)}^2 > \eps^2 \Big)\\
		\le& 
		\frac{1}{\eps^2\NN^2} (
		\normone{\tildeE} + \normone{ E'}
		) O(1) \\
		\le& \frac{1}{\eps^2\NN^2} (N\dmax + N\dmax^2) O(1)\\
		=& \frac{1}{\eps^2\NN} (o(\NN^{1/4}) + o(\NN^{1/2})) O(1)\\
		=& o(1)
	\end{array}
\end{displaymath}
due to Assumption~\ref{assumpt:degree}. Therefore, we have established the statement of the present lemma because we have verified that all three sets of summands in~\eqref{eq:toBound4} are of order $o_{\PP}(1)$. 
\end{proof} 

\begin{proof}[Proof of Theorem~\ref{thm:var-est}]
The proof follows from Lemma~\ref{lem:consistent-variance-part-I} and~\ref{lem:consistent-variance-part-II}. 
\end{proof}

\section{Extension to Estimate Global Effects}\label{sect:extension}

So far, we focused on the EATE.  
We intervened on each individual unit and left the treatment selections of the other units as they were. 

Subsequently, we consider another type of treatment effect where we assess the effect of a single intervention that intervenes on all subjects simultaneously. Instead of the EATE in~\eqref{eq:thetaN}, we subsequently consider the global average treatment effect (GATE) with respect to the binary vector $\pi\in\{0, 1\}^{\NN}$ of treatment selections 
\begin{equation}\label{eq:thetaNglobal}
	\thetazeroNGlobal(\pi)
	= \frac{1}{\NN}\sum_{i=1}^{\NN} \E\Big[\Yi^{do(\W = \pi)} - \Yi^{do(\W = 1-\pi)}\Big],
\end{equation}
where $\W=(W_1,\ldots, W_{\NN})$ denotes the complete vector of treatment selections of all units. In practice, the most common choice is where all components of $\pi$ equal $1$. That is, the treatment effect comes from comparing the situation where all units are assigned to the treatment versus where no-one gets the treatment. 

We use the same definition for $\Si$, $i\in\indset{\NN}$ as before and denote the dependency graph on $\Si$, $i\in\indset{\NN}$ by $\tildeG=(V, \tildeE)$. Furthermore, we let $\alpha(i) = \{j\in\indset{\NN}\colon \{i, j\}\in\tildeE\}\cup\{i\}$ for $i\in\indset{\NN}$ denote the nodes that share an edge with $i$ in the dependency graph together with $i$ itself. For some real number $\xi\in\R$ and a nuisance function triple $\eta=(\gone,\gzero,\h)$, consider the score function
\begin{equation}\label{eq:scoreXi}
	\begin{array}{rcl}
	\psi(\Si, \theta, \xi) 
	&=&\gone(\CXi) - \gzero(\CXi)	+\Big(\prod_{j\in\alpha(i)} \frac{\Wj}{\h(\CZj)}\Big) \big(\Yi - \gone(\CXi)\big)\\
	&&\quad
	- \Big(\prod_{j\in\alpha(i)}\frac{1 - \Wj}{1 - \h(\CZj)}\Big) \big(\Yi - \gzero(\CXi) \big)
	- \xi. 
\end{array}
\end{equation}

In contrast to the score
that we used for the EATE, this score includes additional factors $\frac{\Wj}{\h(\CZj)}$ and $\frac{1 - \Wj}{1 - \h(\CZj)}$ for units $j$ that share an edge with $i$ in the dependency graph. With the GATE, when we globally intervene on all treatment selections at the same time, this also influences the $\Xi$ that are present in $\gone$ and $\gzero$. 
In the score~\eqref{eq:scoreXi}, the ``correction terms'' $(\prod_{j\in\alpha(i)} \frac{\Wj}{\h(\CZj)}) (\Yi - \gone(\CXi))$ and $(\prod_{j\in\alpha(i)}\frac{1 - \Wj}{1 - \h(\CZj)}) (\Yi - \gzero(\CXi) )$ are only active if $i$ and the units from which it receives spillover effects have the same observed treatment selection. 

Let us denote by 
\begin{displaymath}
	\thetazeroiGlobal =  
	\E\Big[\Yi^{do(\W = \pi)} - \Yi^{do(\W = 1-\pi)}\Big] = 
	\E\big[\gonezero(\CXione) - \gzerozero(\CXizero)\big]
\end{displaymath} 
the $i$th contribution in~\eqref{eq:thetaNglobal}. 
Here, 
\begin{displaymath}
	\Xipi = \big(f^1_x(\pi_{-i}, \Cminusi, A), \ldots, 
	f^r_x(\pi_{-i}, \Cminusi, A) 
	\big)
\end{displaymath}
denotes the feature vector where $\Wj$ is replaced by $\pi_j$, and 
\begin{displaymath}
	\Ximinuspi = \big(
 f^1_x(1-\pi_{-i}, \Cminusi, A), \ldots,
	f^r_x(1-\pi_{-i}, \Cminusi, A) 
	\big)
\end{displaymath}
denotes the feature vector where $\Wj$ is replaced by $1 - \pi_j$. The features $\Zipi$ and $\Ziminuspi$ are defined analogously. 
Similarly to Lemma~\ref{lem:identifiability}, it can be shown that $\E[\psi(\Si,\thetazeroiGlobal,\etazero)] = 0$ holds, which lets us identify the global treatment effect $\thetazeroNGlobal$ by
\begin{displaymath}
	\thetazeroNGlobal = \frac{1}{\NN}\sum_{i=1}^{\NN} \E[\phi(\Si,\etazero)], 
\end{displaymath} 
where 
\begin{displaymath}
	\begin{array}{rcl}
	\phi(\Si, \eta) 
	&=&\gone(\CXi) - \gzero(\CXi)	+\Big(\prod_{j\in\alpha(i)} \frac{\Wj}{\h(\CZj)}\Big) \big(\Yi - \gone(\CXi)\big)\\
	&&\quad
	- \Big(\prod_{j\in\alpha(i)}\frac{1 - \Wj}{1 - \h(\CZj)}\Big) \big(\Yi - \gzero(\CXi) \big). 
\end{array}
\end{displaymath}

To estimate $\thetazeroNGlobal$, we  apply the same procedure as for the ATE. The only difference is that when we evaluate the machine learning estimates, we do not use the observed treatment selections, but instead insert the respective components of $\pi$ and $1-\pi$. However, we insert the actually observed treatment selections in the product terms $\prod_{j\in\alpha(i)} \frac{\Wj}{\h(\CZj)}$ and $\prod_{j\in\alpha(i)}\frac{1 - \Wj}{1 - \h(\CZj)}$. 
This gives the estimator $\hthetaGlobal$. Analogously to Theorem~\ref{thm:Gaussian} for the EATE, also the GATE with respect to $\pi$ converges at the parametric rate and follows a Gaussian distribution asymptotically. 

\begin{theorem}[Asymptotic distribution of $\hthetaGlobal$]\label{thm:GaussianGlobal}
	Assume Assumption~\ref{assumpt:regularity} (with $\theta$ replaced by $\xi$), 
	\ref{assumpt:degree}, and~\ref{assumpt:DML} in the appendix in Section~\ref{sect:AssumptionsDefinitions} hold. Furthermore, assume that there exists a finite real constant $L$ such that $\normone{\alpha(i)}\le L$ holds for all $i\in\indset{\NN}$. 
	
	Then, the estimator $\hthetaGlobal$ of the GATE with respect to $\pi\in\{0,1\}^{\NN}$, $\thetazeroNGlobal$, satisfies
	\begin{displaymath}\label{eq:GaussThm}
		\sqrt{\NN}(\hthetaGlobal-\thetazeroNGlobal) \stackrel{d}{\rightarrow} \mathcal{N}(0, \sigmainfty), 
	\end{displaymath}
	where $\sigmainfty$ is characterized in Assumption~\ref{assumpt:variance} with the $\psi$ in~\eqref{eq:scoreXi}. 
	The convergence in~\eqref{eq:GaussThm} is in fact uniformly over the law $P$ of the observations. 
\end{theorem}

This theorem requires that the number of spillover effects a unit receives is bounded. 
Theorem~\ref{thm:Gaussian} that establishes the parametric convergence rate and asymptotic Gaussian distribution of the EATE estimator did not require such an assumption. 
The reason is that $\hzero(\Ci, \Zi)$ represents the conditional expectation of $\Wi$ given $\Ci$ and $\Zi$ and consequently a probability taking values in the interval $(0, 1)$. If we allowed $\normone{\alpha(i)}$ to grow with $\NN$, the products $\prod_{j\in\alpha(i)} \frac{\Wj}{\h(\CZj)}$ and $\prod_{j\in\alpha(i)}\frac{1 - \Wj}{1 - \h(\CZj)}$ would diverge. 

To estimate $\sigmainfty^2$ in Theorem~\ref{thm:GaussianGlobal}, we can apply the procedure described in Section~\ref{sect:var_est}, where we replace $\psi$, $\phi$, and the point estimators by the respective new quantities. 
Also an  analogue
of Theorem~\ref{thm:var-est} holds, but where we assume the setting of Theorem~\ref{thm:GaussianGlobal} holds and that $\normone{\Acald}\to\infty$ as $\NN\to\infty$ for all $d\ge 0$. In particular, we do not require Assumption~\ref{assumpt:Acald} and~\ref{assumpt:kappa} formulated in the appendix in Section~\ref{sect:AssumptionsDefinitions}.  Furthermore, to prove consistency of the variance estimator, it is sufficient to establish that the degree-specific causal effect estimators $\hthetadGlobal$, which are defined analogously to $\hthetad$, are consistent. In particular, they are not required to converge at a particular rate.  

Also~\citet{Laan2014}, \citet{laan2017}, and \citet{sofrygin2017} consider semiparametric estimation of the GATE using TMLE. They also require a uniform bound of the number of spillover effects a unit receives to achieve the parametric convergence rate of their estimator.

\end{appendices}

\end{document}